\documentclass{iopconfser}

\usepackage[ansinew]{inputenc}
\usepackage[bbgreekl]{mathbbol}
\usepackage{amsmath,amssymb,amsthm}
\usepackage{amsfonts}
\usepackage{bm}
\usepackage{bbm}
\usepackage{xcolor}

\usepackage{graphicx}    
\usepackage{subcaption}  

\newcommand{\fracd}[2]{\ensuremath{\frac{\delta   #1}{\delta   #2}}}
\newcommand{\fracp}[2]{\ensuremath{\frac{\partial #1}{\partial #2}}}
\newcommand{\fract}[2]{\ensuremath{\frac{\mathrm{d} #1}{\mathrm{d} #2}}}
\newcommand{\dd}{\mathop{}\!\mathrm{d}} 
\newcommand{\ee}{\mathop{}\!\mathrm{e}}
\newcommand{\ii}{\mathop{}\!\mathrm{i}}

\newcommand{\half}{\frac{1}{2}}
\newcommand\bk{\boldsymbol{k}}
\newcommand\bX{\boldsymbol{X}}
\newcommand\bbX{\bar{\boldsymbol{X}}}
\newcommand\bV{\boldsymbol{V}}
\newcommand\bbV{\bar{\boldsymbol{V}}}
\newcommand\bx{\boldsymbol{x}}
\newcommand\bv{\boldsymbol{v}}
\newcommand\bE{\boldsymbol{E}}
\newcommand\bB{\boldsymbol{B}}
\newcommand\bA{\boldsymbol{A}}
\newcommand\bJ{\boldsymbol{J}}
\newcommand\bD{\boldsymbol{D}}
\newcommand\bH{\boldsymbol{H}}

\newcommand\cH{\mathcal{H}}
\newcommand\cR{\mathcal{R}}

\newcommand{\arr}[1]{{\bm{\mathsf{#1}}}}

\newcommand\arrA{{\bm{\mathsf{A}}}}
\newcommand\arrB{{\bm{\mathsf{B}}}}

\newcommand\arrD{{\bm{\mathsf{D}}}}
\newcommand\arrE{{\bm{\mathsf{E}}}}
\newcommand\arrF{{\bm{\mathsf{F}}}}
\newcommand\arrH{{\bm{\mathsf{H}}}}

\newcommand\arrG{{\bm{\mathsf{G}}}}
\newcommand\arrJ{{\bm{\mathsf{J}}}}
\newcommand\arrM{{\bm{\mathsf{M}}}}
\newcommand\arrP{{\bm{\mathsf{P}}}}

\newcommand\arrV{{\bm{\mathsf{V}}}}

\newcommand\arrX{{\bm{\mathsf{X}}}}

\newcommand\arrphi{{\bm{\mathsf{\phi}}}}
\newcommand\arrrho{{\bm{\mathsf{\rho}}}}
\newcommand\arrpsi{{\bm{\mathsf{\psi}}}}

\def\matd{\mathbb{d}}

\def\matC{\mathbb{C}}
\def\matD{\mathbb{D}}

\def\matG{\mathbb{G}}
\def\matH{\mathbb{H}}

\def\matI{\mathbb{I}}

\def\matO{\mathbb{O}}

\newcommand{\Xd}{\ensuremath{\dot{\bX}}}
\newcommand{\vpar}{\ensuremath{{v_\shortparallel}}}
\newcommand{\Vpar}{\ensuremath{{V_\shortparallel}}}

\newcommand{\Aex}{\ensuremath{\boldsymbol{A}_{\text{ext}}}}

\newcommand{\bex}{\ensuremath{\boldsymbol{b}_{\text{ext}}}}
\newcommand{\bexp}{\ensuremath{\boldsymbol{b}_{\text{ext}}(\bbX_p)}}
\newcommand{\Bex}{\ensuremath{\boldsymbol{B}_{\text{ext}}}}

\newcommand{\Btot}{\ensuremath{\boldsymbol{B}_{\text{tot}}}}

\newcommand{\Bpartot}{\ensuremath{B_{\shortparallel,\text{tot}}}}
\newcommand{\Bpartotp}{\ensuremath{B_{\shortparallel,\text{tot},p}}}

\newcommand{\Bexm}{\ensuremath{\left|\Bex \right|}}

\newcommand{\Bpar}{\ensuremath{B_\shortparallel}}
\newcommand{\Bperp}{\ensuremath{\bB_\perp}}
\newcommand{\Bperpm}{\ensuremath{\left|\bB_\perp \right|}}
\newcommand{\Eperp}{\ensuremath{\bE_\perp}}

\newcommand{\Eperpm}{\ensuremath{\left|\bE_\perp \right|}}
\newcommand{\Astar}{\ensuremath{\boldsymbol{A}^\ast}}
\newcommand{\Bstar}{\ensuremath{\boldsymbol{B}^\ast}}
\newcommand{\boldsymboltar}{\ensuremath{\boldsymbol{B}^\ast}}
\newcommand{\boldsymboltarp}{\ensuremath{\boldsymbol{B}^\ast(\bbX_p)}}
\newcommand{\Bpstar}{\ensuremath{{B_\shortparallel^\ast}}}
\newcommand{\Bpstarp}{\ensuremath{{B_{\shortparallel}^\ast(\bbX_p)}}}

\newcommand{\vgc}{\ensuremath{\boldsymbol{v}_{gc}}}
\newcommand{\agc}{\ensuremath{a_{gc}}}
\newcommand{\jgc}{\ensuremath{\boldsymbol{J}_{gc}}}
\newcommand{\jgcpar}{\ensuremath{J_{\shortparallel}}}

\newcommand{\rhogc}{\ensuremath{\rho_{gc}}}

\newcommand{\kpar}{k_\shortparallel}
\newcommand{\kperp}{k_\perp}

\newcommand{\tA}{\ensuremath{\delta{\boldsymbol{A}}}}

\newcommand{\tphi}{\ensuremath{\delta{\phi}}}

\newtheorem{proposition}{Proposition}
\newtheorem{remark}{Remark}

\begin{document}
\title{A geometric Particle-In-Cell discretization of the drift-kinetic and fully kinetic Vlasov-Maxwell equations}
\author{Guo Meng$^{1*}$, Katharina Kormann$^3$, Emil Poulsen$^1$, Eric Sonnendr\"ucker$^{1,2}$}

\affil{$^1$ Max Planck Institute for Plasma Physics, Garching, Germany}

\affil{$^2$ School of Computation Information and Technology, Technical University of Munich, Garching, Germany}

\affil{$^3$ Department of Mathematics, Ruhr University Bochum, Bochum, Germany}

\begin{abstract}
    In this paper, we extend the geometric Particle in Cell framework on dual grids to a gauge-free drift-kinetic Vlasov--Maxwell model and its coupling with the fully kinetic model.
  We derive a discrete action principle on dual grids for our drift-kinetic model, such that the dynamical system involves only the electric and magnetic fields and not the potentials as most drift-kinetic and gyrokinetic models do. This yields a macroscopic Maxwell equation including polarization and magnetization terms that can be coupled straightforwardly with a fully kinetic model.
\end{abstract}

\renewcommand{\thefootnote}{*} %
\footnotetext{Corresponding author: guo.meng@ipp.mpg.de}

\section{Introduction}
Many plasma physics models have been proved to possess a Hamiltonian structure \cite{morrison1998hamiltonian, morrison2017structure}. Invariants, like the Hamiltonian or Casimir invariants like the Gauss law and $\text{div} B = 0$ emerge naturally within this framework. Adequate conservation of these quantities has been proven to be essential for well-behaving numerical solutions. General structure preserving numerical methods aim at preserving one or more of these invariants.
Geometric numerical methods achieve this by discretizing the Hamiltonian structure, i.e. Poisson bracket and Hamiltonian or an action principle, rather than the resulting partial differential equations (PDEs). This approach approximates the infinite dimensional original Hamiltonian system by a finite-dimensional Hamiltonian system and in this way guarantees the conservation of the appropriately discretized invariants.
This can be achieved naturally by using compatible discretizations of the fields following a discrete de Rham sequence and a particle description of the Vlasov equation. Compatible discretizations of the fields can be achieved with appropriate choices of the Finite Element spaces (cf.~\cite{kraus2017gempic,kormann2021energy,perse2021geometric,campos2022variational}). However, these have the drawback that they involve the inversion of a mass matrix at every time step even for explicit schemes, which is not optimal for high-performance computing (HPC) implementations. For this reason, a new geometric discretization related to Mimetic Finite Differences was developed in \cite{Kormann2024A-Dual-Grid}.

The aim of this paper is to extend the geometric particle in cell (PIC) framework to gyrokinetic models in the Zero Larmor Radius limit, which corresponds to the well-known drift-kinetic model \cite{kulsrud1983mhd}, and couple them seamlessly with the fully kinetic model. This is achieved best by considering gauge-free gyrokinetic models, which means that the corresponding dynamical system can be expressed only using the electric and magnetic fields and not the potentials. The final models are derived from a Lagrangian possibly involving several species of drift-kinetic particles as well as several species of fully kinetic particles. We can of course have all particles drift-kinetic or all particles fully kinetic. We call hybrid the model where both particle types are present.

Starting from the gyrokinetic Lagrangian of Burby and Brizard \cite{burby2019gauge-free}  in the Zero Larmor Radius limit, we introduce a Lagrangian coupling drift-kinetic electrons with fully kinetic ions at the continuous level and then propose a discretized Lagrangian based on the Mimetic Finite Difference framework on dual grids \cite{Kormann2024A-Dual-Grid} and the PIC method, and derive from it the equations of motions for the PIC markers as well as the discrete generalized Maxwell equations involving the polarization and magnetization terms coming from the drift-kinetic particles. The new method is verified by checking the wave spectrum for only fully kinetic, only drift-kinetic and coupled models. This comparison also highlights the differences between these models. Notice however that we don't add any approximation beyond the gyrokinetic particle Lagrangian. This means in particular that light waves and compressional Alfv\'en waves are still present in our drift-kinetic models. Addressing Darwin like and quasi-neutrality assumptions to remove these high frequency waves will be the purpose of future work.

Kinetic equations for a collisionless magnetized plasma in the low-frequency limit and the derivation of MHD equations are discussed in Ref. \cite{kulsrud1983mhd}.  In recent years, kinetic codes have been developed to model the edge and scrape-off layer (SOL) regions of magnetic confinement fusion devices \cite{michels2021geneX,boesl2019picls,dorf2016cogent}, as the validity of fluid models in these regions remains an open question. For the fully nonlinear collision operator \cite{hirvijoki2021structure,lu2024fullfcoll} required in edge plasma physics, a fully kinetic treatment of ions provides a natural and rigorous solution, while a drift-kinetic model for electrons is generally sufficient for capturing edge plasma behavior \cite{chankin2012kipp}. Moreover, the fully kinetic method facilitates coupling between ion plasma and neutral particles, the latter of which can not be described by gyrokinetic theory at all.

Notably, Refs. \cite{YChen2009VlasovDrift} and \cite{chen2019gefiEB} introduce a hybrid kinetic model using the E\&B formulation directly. As discussed in \cite{YChen2009VlasovDrift,rosen2022EB}, the compressional Alfv\'en mode imposes the most stringent constraint on the time step when $k_\theta\rho_i \sim 1$, advanced numerical methods such as the implicit methods \cite{lu2021implicit} are necessary for overcoming this constraint. The gyrokinetic E\&B models developed in previous works for low-frequency electromagnetic fluctuations \cite{chen2021gkEB}, kinetic Alfv\'en waves in tokamak plasmas \cite{rosen2022EB} and hybrid model \cite{YChen2009VlasovDrift, chen2019gefiEB} differ from our approach in both the formulations and the focus on geometric numerical methods for discretizing the equations. A key feature of our approach is that we derive both the field and particle equations of motion from a Lagrangian coupling drift-kinetic electrons with fully kinetic ions at the continuous level, ensuring numerical consistency and structure preserving in the discretized space. This contrasts with the treatment that separately discretizes the physical field and particle equations, for which maintaining the same conservation laws as in the continuous Lagrangian can be challenging.
In Ref. \cite{burby2019gauge-free} Burby \& Brizard derive a new gauge-free formulation of gyrokinetic theory but do not consider its discretization. We start from their model in the Zero Larmor radius limit and show how a structure-preserving numerical scheme can be derived from it. In our work we also consider coupled drift-kinetic and fully kinetic models. Specifically, we ensure that our formulation naturally leads to a self-consistent coupling between the drift-kinetic and fully kinetic models within a hybrid framework. Furthermore, our model allows easy switching between kinetic models for different species and straightforward coupling of various kinetic species, which offers new advantages for addressing multiscale physics in plasma systems. 

The paper is organized as follows. In Section \ref{sec:DKmodel} we introduce the gauge-free drift-kinetic model that we are going to discretize. In Section \ref{sec:DKFK} we introduce the coupling of our drift-kinetic and the fully kinetic models. In Section \ref{sec:mimeticFD}
we present the new geometric PIC discretization. 
In Section \ref{sec:timeDisc}, we describe the time discretization and time-stepping methods for the Vlasov-Maxwell system. In Section \ref{sec:dispersion}, we derive the dispersion relation for the drift-kinetic model in a slab geometry. In Section \ref{sec:results} we verify that our models agree with the dispersion relations and highlight the different waves that exist for each model. Finally, in Section \ref{sec:conclusion} we summarize the results and provide some outlook for future work.

\section{\label{sec:DKmodel}A gauge-free drift-kinetic model}\label{sec:GaugeFreeDK}
\subsection{Action principle}\label{sec:ActionDK}

We consider here the gauge-free formulation of the gyrokinetic model introduced in \cite{burby2019gauge-free} in the zero Larmor radius (or long wavelenth) limit, namely the drift-kinetic model. A dispersion relation for this model has been derived in \cite{Zonta2021Dispersion}.

Let us denote by $\Bex$ the external equilibrium magnetic field, by $\bE$ and $\bB$ the perturbed electric and magnetic fields, by $\phi$ the electrostatic potential and by $q$ the charge of the particle species. We define the guiding center single particle Hamiltonian by $H=q\phi + K$ and we split the kinetic energy $K$ into three parts: the term $K_0$ independent of $\bE$ and $\bB$, $K_1$ depending linearly on $\bE$ and $\bB$, and $K_2$ depending quadratically on $\bE$ and $\bB$, i.e.,
\begin{align}
K_0 &= \frac{1}{2}m \vpar^2 + \mu \Bexm, \label{def:K0}\\
K_1 &= \mu \bex \cdot \bB,  \label{def:K1}\\
K_2 &=  \left( \mu \Bexm - m \vpar^2 \right) \frac{ \Bperpm^2 }{2 \Bexm^2} - \frac{m \Eperpm^2}{2 \Bexm^2} - \frac{m \vpar \bE \times \bex \cdot \bB}{ \Bexm^2}, \label{def:K2}
\end{align}
where $\vpar$ is the velocity parallel to external magnetic field, $\vpar = \bv \cdot \Bex$. Here, we have introduced the notation $\Bperp = (Id - \bex \bex^\top)\bB$, $\Eperp = (Id - \bex \bex^\top)\bE$ and $\bex = \frac{\Bex}{\Bexm}$.
$K_0$ is the kinetic energy of the guiding centers, $K_1$ and $K_2$ correspond to the zero Larmor radius limit of the model derived by Burby and Brizard \cite{burby2019gauge-free} in the case of a uniform $\Bex$. The term $K_2$ also appears in Eq.~(54) in Brizard and Hahm \cite{brizard2007foundations}.
 Let us observe that
\begin{equation}\label{eq:K0pK1}
	K_0+ K_1 =  \frac{1}{2}m \vpar^2 + \mu \bex\cdot ( \Bex + \bB).
\end{equation}

Let us also introduce the following notations
\begin{equation}\label{eq:Astar}
	\Astar= \bA + \Aex +\frac mq \vpar \bex,~~ \Bstar= \nabla\times\Astar = \bB +\Bex + \frac mq \vpar \nabla\times\bex, 
	~~ \Bpstar = \Bstar\cdot \bex.
\end{equation}

Denoting  the guiding center distribution function by $f$ and using the guiding center volume element $\Bpstar \dd\bx\dd \vpar  \dd\mu$ expressed at the initial time, the  Lagrangian is defined as, assuming implicitly a sum over all the particle species in the particle part
\begin{multline}\label{eq:action}
	L(\bA(t),\dot{\bA}(t), \phi(t), \bX(t),\Xd(t),\vpar(t), \mu; f_0) \\
	= \int  q \left(\Astar(t,\bX(t),\Vpar(t),\bA(t,\bX(t)))  \cdot \Xd(t) - q \phi(t,\bX(t))\right) f_0 \Bpstar_0 \dd \bx_0  \dd \vpar_0 \dd \mu \\
	- \int \left( K_0(\bX(t), \Vpar(t),\mu)+ (K_1+K_2)(\bX(t), \Vpar(t),\bE(t, \bX(t)), \bB(t,\bX(t));\mu)  \right)f_0 \Bpstar_0 \dd \bx_0 \dd \vpar_0 \dd \mu \\
	 + \int \frac{1}{2} \left( \varepsilon_0 \left| \bE(t,\bx)\right|^2 - \frac{1}{\mu_0} \left| \Bex(\bx) + \bB(t,\bx)\right|^2 \right) \dd \bx 
	\end{multline}
where $\bE = - \fracp{\bA}{t} - \nabla\phi$ and $\bB = \nabla\times\bA$. 

\subsection{The dynamical system}
The equations of motion can be obtained by setting to zero the variations with respect to each of the time-dependent variables in the Lagrangian.
First, the variations with respect to $\bX$ classically yield the Euler--Lagrange equation 
$\fracp{}{t} \fracp{L}{\Xd} = \fracp{L}{\bX}$.
As $K$ depends also on $\bX$ through $\bE$ and $\bB$ we shall denote by 
$\fract{K}{\bX}=\nabla K + \nabla\bE \fracd{K}{\bE} + \nabla\bB \fracd{K}{\bB}$, the total derivative, where $\nabla K$ is the partial derivative with respect to $\bX$. 
In our case we have $ \fracp{L}{\Xd} = q \Astar$ and $\fracp{L}{\bX} = q((\nabla\Astar) \Xd - \nabla\phi)-\fract{K}{\bX}$, so that the Euler--Lagrange equation becomes
\begin{equation}\label{eq:EL-X-0}
	q\left(\fracp{\bA}{t} + \nabla\phi + \frac mq \dot{\Vpar}\bex - \Xd\times\nabla\times\Astar\right) = -\fract{K}{\bX}
\end{equation}
or equivalently
\begin{equation}\label{eq:EL-X}
	m \dot{\Vpar}\bex - q \Xd\times\Bstar = q\bE - \fract{K}{\bX}.
\end{equation}
Now for the Euler--Lagrange equation in $\Vpar$, we have $ \fracp{L}{\dot{\Vpar}} = 0$ and
$ \fracp{L}{\Vpar} = q\fracp{\Astar}{\Vpar}\cdot\Xd - \fracp{K}{\Vpar}$ so that
\begin{equation}\label{eq:EL-Vpar}
	 m \bex\cdot\Xd  =\fracp{K}{\Vpar}.
\end{equation}

Let us now solve Eqs. \eqref{eq:EL-X}--\eqref{eq:EL-Vpar} for $\Xd$ and $\dot{\Vpar}$. We first take the cross product of Eq. \eqref{eq:EL-X} with $\bex$
$$(\Bstar\cdot\bex)\Xd - (\bex\cdot\Xd)\Bstar =  (\bE - \frac 1q \fract{K}{\bX})\times\bex$$
which yields the following equation for $\Xd$ 
\begin{equation}\label{eq:XdotAbs}
 \Xd = \frac{1}{\Bpstar}\left(\frac 1 m \fracp{K}{\Vpar}\Bstar +  (\bE - \frac 1 q \fract{K}{\bX})\times\bex \right) =: \vgc
\end{equation}
where this expression defines the guiding center velocity $\vgc$.
And taking the dot product of Eq. \eqref{eq:EL-X} with $\Bstar$ yields
\begin{equation}\label{eq:VpdotAbs}
	\dot{\Vpar} =  \frac{1}{\Bpstar}\left(\frac q m \bE - \frac 1 m \fract{K}{\bX}\right) \cdot\Bstar =: \agc,
\end{equation}
where this expression defines the guiding center parallel acceleration $\agc$.

\begin{proposition}
	The density $\Bpstar$ verifies the Liouville equation for the characteristics Eqs. \eqref{eq:XdotAbs}--\eqref{eq:VpdotAbs}
	\begin{equation}\label{eq:liouville}
	\fracp{\Bpstar}{t} + \nabla\cdot (\Bpstar \Xd ) + \fracp{\Bpstar \dot{\Vpar} }{\Vpar} =0,
	\end{equation}
	which implies that $\Bpstar f$ is conserved along these characteristics for any smooth function $f$.
\end{proposition}
\begin{proof}
	Using the definitions in Eq. \eqref{eq:Astar} and that $\nabla\cdot\Bstar=0$
	we compute
	\begin{align}
		\fracp{\Bpstar}{t} &= \fracp{\bB}{t}\cdot\bex, \\
		\nabla\cdot (\Bpstar \Xd) &=  \frac 1 m \nabla\cdot(\fracp{K}{\Vpar}\Bstar) + \nabla\cdot(\bE\times\bex) - \frac 1q \nabla\cdot(\fract{K}{\bX} \times\bex), \nonumber\\ 
		&= \frac 1 m \nabla\fracp{K}{\Vpar}\cdot \Bstar + \bex\cdot (\nabla\times\bE) -\bE\cdot(\nabla\times\bex)  + \frac 1q \fract{K}{\bX}\cdot(\nabla\times\bex), \\ 
		\fracp{\Bpstar \dot{\Vpar} }{\Vpar} &= (\bE - \frac 1 q \fract{K}{\bX}) \cdot \nabla\times\bex -\frac 1m \nabla \fracp{K}{\Vpar}\cdot\Bstar.
	\end{align}
We get the result by summing these three equations using Faraday's law $\fracp{\bB}{t}  + \nabla\times\bE=0$ which follows from the definition of the potentials.	

\end{proof}

Let us observe that the potentials appear in the Lagrangian through the fields $\bE = -\fracp{\bA}{t}-\nabla\phi$
and $\bB =\nabla\times\bA$ from which it follows that
\begin{align}
\int \fracd{\mathcal{F}[\bE]}{\bA}\cdot\tA \dd \bx &= -\int \fracd{\mathcal{F}}{\bE} \cdot \fracp{\tA}{t}\dd \bx  , \\
\int \fracd{\mathcal{F}[\bE]}{\phi}\cdot\tphi \dd \bx &= - \int \fracd{\mathcal{F}}{\bE} \cdot \nabla\tphi \dd \bx , \\
\int \fracd{\mathcal{F}[\bB]}{\bA}\cdot \tA \dd \bx &= \int \fracd{\mathcal{F}}{\bB} \cdot \nabla\times\tA \dd \bx .
\end{align}

We can now compute the variations with respect to $\bA$ (integrating by parts in time):
\begin{multline}\label{eq:var-A}
	0= \int  q \tA  \cdot \Xd(t)  f_0 \Bpstar_0\dd \bx_0  \dd \vpar_0 \dd \mu \dd t
	 - \int \left(\fracp{}{t}\fracd{ K}{\bE}\cdot \tA  +\fracd{K}{\bB}\cdot\nabla\times\tA\right) f_0 \Bpstar_0\dd \bx_0  \dd \vpar_0 \dd \mu \dd t\\
	+ \int  \left( \varepsilon_0  \fracp{\bE}{t} \cdot \tA - \frac{1}{\mu_0} ( \Bex(\bx) + \bB(t,\bx)) \cdot\nabla\times \tA \right) \dd \bx \dd t
\end{multline}

Hence, making the change of variables $(\bx, \vpar)= (\bX, \Vpar) (t;\bx_0,\vpar_0)$ using
that $f\Bpstar$ is conserved by this transformation due to Liouville's theorem and $\vgc=\Xd$, we can define the guiding center current
\begin{equation}\label{eq:jgy}
\jgc = q \int   \vgc  f \Bpstar \dd \vpar \dd \mu ,
\end{equation}

as well as the polarization $ \mathbf{P}$ and the magnetization $\mathbf{M}$
\begin{align}
	\mathbf{P} &= -\int \fracd{ K}{\bE}  f \Bpstar  \dd \vpar \dd \mu, \label{def:P}\\
	\mathbf{M} &=  -\int \fracd{K}{\bB} f \Bpstar  \dd \vpar \dd \mu . \label{def:M}
\end{align}
On the other hand, we define the displacement field $\bD$ and the magnetic field intensity $\bH$ by
\begin{align}
	\bD &= \epsilon_0 \bE + \mathbf{P}, \label{def:D}\\
	\bH &= \frac{1}{\mu_0}(\Bex+\bB) - \mathbf{M} \label{def:H}
\end{align}
to finally be able to write the weak form of Ampere's law
\begin{equation}\label{eq:ampere-weak}
	 \int  \left( \fracp{\bD}{t} \cdot \tA -  \bH \cdot\nabla\times \tA \right) \dd \bx 
	= - \int \jgc \cdot \tA \dd \bx \quad\forall \tA .
\end{equation}

Let us finally set the variation with respect to $\phi$ to zero:
\begin{equation}\label{eq:var-rho}	
   \int \fracd{K}{\bE} \cdot \nabla\tphi  f \Bpstar \dd \bx  \dd \vpar \dd \mu 
	- \epsilon_0 \int \bE\cdot\nabla\tphi \dd \bx =\int  q \tphi   f \Bpstar \dd \bx  \dd \vpar \dd \mu \quad\forall \tphi.
\end{equation}
Defining the guiding center charge density $\rhogc$ by 
\begin{equation}\label{eq:rhogc}
 \rhogc  = q \int   f \Bpstar   \dd \vpar \dd \mu
\end{equation}
this becomes the weak Poisson equation using the definition Eq. \eqref{def:D} of $\bD$
\begin{equation}\label{eq:poisson-weak}
	-\int  \bD\cdot\nabla\tphi \dd \bx 
   = \int \rhogc  \tphi \dd \bx \quad\forall \tphi .
\end{equation}

In addition, the definition of the fields $\bE$ and $\bB$ yields Faraday's law $\fracp{\bB}{t}+\nabla\times\bE=0$ and $\nabla\cdot\bB =0$ so that we get the following guiding center Maxwell equations (in strong form)
\begin{align}
\fracp{\bD}{t} - \nabla \times \bH &=- \jgc , \label{eq:ampere}\\
\fracp{\bB}{t} + \nabla \times \bE &= 0,\label{eq:faraday}\\
\nabla \cdot \bD &= \rhogc , \label{eq:Gauss}\\
\nabla \cdot  \bB &= 0.
\end{align}

These are coupled with the equations of motion \eqref{eq:XdotAbs}--\eqref{eq:VpdotAbs} that we recall for convenience
\begin{equation}\label{eq:Xdot-2}
	\fract{\bX}{t} = \frac{1}{\Bpstar}\left(\frac 1 m \fracp{K}{\Vpar}\Bstar +  (\bE - \frac 1 q \fract{K}{\bX})\times\bex \right) = \vgc,
   \end{equation}
\begin{equation}\label{eq:Vdot-2}
	   \fract{\Vpar}{t} =  \frac{1}{\Bpstar}\left(\frac q m \bE - \frac 1 m \fract{K}{\bX}\right) \cdot\Bstar = \agc .
\end{equation}

\begin{remark}
	Integrating the Liouville equation \eqref{eq:liouville} over $\vpar$ and $\mu$ yields the guiding center continuity equation
	\begin{equation}\label{eq:continuity-dk}
	\fracp{\rhogc}{t}+ \nabla\cdot\jgc =0.
	\end{equation}
	This implies in particular that Eq. \eqref{eq:Gauss} remains satisfied at all times provided it is satisfied at the initial time and Ampere's equation \eqref{eq:ampere} is satisfied. This is similar to the classical Vlasov-Maxwell equations.
\end{remark}

The dynamical system of Eqs. \eqref{eq:ampere}--\eqref{eq:Vdot-2} forms a noncanonical Hamiltonian system coupling particles and fields.
Due to the time derivative in the definition of $\bE$ from $\bA$ the Hamiltonian associated to the Lagrangian Eq. \eqref{eq:action} is
\begin{align}\label{eq:ham}
	\cH &= \frac{ \varepsilon_0}{2} \int \left| \bE\right|^2 \, \dd \bx+ \frac{1}{2\mu_0} \int \left| \bB + \Bex\right|^2 \, \dd \bx + \int  \left( K -  \bE \cdot \frac{\delta K}{\delta \bE} \right)
	f_0 \Bpstar_0 \dd \bx_0 \dd \vpar_0 \dd \mu  \nonumber \\
	&= \frac{ \varepsilon_0}{2} \int \left| \bE\right|^2 \, \dd \bx+  \int \bE\cdot \mathbf{P} \dd \bx + \frac{1}{2\mu_0} \int \left| \bB + \Bex\right|^2 \, \dd \bx 
	+ \int  K f_0 \Bpstar_0 \dd \bx_0 \dd \vpar_0 \dd \mu 
\end{align}

\subsection{\label{sec:linearized} Linearized polarization and magnetization}

Let us first compute the explicit expressions of the polarization $ \mathbf{P}$ and magnetization $ \mathbf{M}$ for the specific kinetic energy defined by Eqs. \eqref{def:K0}--\eqref{def:K2}.
We can compute $ \mathbf{P}$  using the definition of Eq. \eqref{def:P} and the expression of $K_2$, given in Eq. \eqref{def:K2}, the only part of $K$ depending on $\bE$
\begin{equation}\label{eq:Pexplicit}
	\mathbf{P} (t,\bx)= -\int \fracd{ K}{\bE}  f \Bpstar  \dd \vpar \dd \mu = \frac{m}{\Bexm^2} \int (\bE_\perp + \vpar \bex\times \bB)  f \Bpstar  \dd \vpar \dd \mu.
\end{equation}
The magnetization $\mathbf{M}$ defined Eq. \eqref{def:M} depends both on $K_1$ and $K_2$. 
It reads
\begin{multline}\label{eq:Mexplicit}
	\mathbf{M} (t,\bx)= -\int \fracd{ K}{\bB}  f \Bpstar  \dd \vpar \dd \mu \\
	= \int \left[ -\mu\bex - \frac{\mu\bB_\perp}{\Bexm} +
		\frac{m}{\Bexm^2}  (\vpar^2\bB_\perp + \vpar \bE\times\bex )  \right]f \Bpstar  \dd \vpar \dd \mu.
\end{multline}
These terms depend on the distribution function, which needs to be calculated from the particles. Moreover, keeping the $K_2$ in the Lagrangian will yield additional contributions to the particles equations of motion. 

For many problems the fully nonlinear polarization and magnetization is not needed, and we can linearize the higher order quadratic term $K_2$ keeping the contribution from $K_1$ nonlinear 
by replacing $f$ by the local Maxwellian $F_M$  with density $n_{M}(\bx)$, parallel velocity  $u_{\shortparallel}$
and thermal velocity $v_{th}=\sqrt{k_B T/m}$, where $k_B$ is the Boltzmann constant and $T$ the temperature of the species:
\begin{equation}\label{eq:Maxwellian}
	f_{M}(\bx,v_\shortparallel,\mu)=\frac{n_{M}(\bx)}{(2\pi)^{3/2} v_{th}^3 m}\exp\left[-\frac{(v_\shortparallel-u_{\shortparallel}(\bx))^2}{2v_{th}(\bx)^2}-\frac{\mu|\Bex(\bx)|}{m v_{th}(\bx)^2}\right],
\end{equation}
normalized so that $\int f_{M} \Bexm \dd \vpar \dd \mu = n_{M}$. 

Then replacing the actual particle distribution function $f$ by the Maxwellian in the expressions above,
except for the $K_1$ term,
we obtain
\begin{align}
	\mathbf{P} (t,\bx) &=  \frac{m \, n_M}{\Bexm^2}(\bE_\perp + u_{\shortparallel,M} \bex\times \bB), \label{eq:Plin}\\ 
	\mathbf{M} (t,\bx) &=   - \int \mu\bex f \Bpstar  \dd \vpar \dd \mu
	+ \frac{n_M}{\Bexm^2}  m u_{\shortparallel,M} \bE\times\bex . \label{eq:Mlin}  
\end{align}
Notice that the terms in $\bB_\perp$ cancel when integrating over $\vpar$ and $\mu$. Although Eq. \eqref{eq:Plin} is derived for the shifted-Maxwellian distribution, the result holds for any distribution, provided that the parallel flow is defined as $n u_\parallel = \int v_\parallel  f \Bpstar  \dd \vpar \dd \mu$. In the gyrokinetic theory, the polarization current is defined as the time derivative of the polarization vector, $\bJ_{\text{pol}}=\partial \mathbf{P}/\partial t$. The first term on the right-hand side of Eq. \eqref{eq:Plin} yields the polarization current $\bJ_{\text{pol}}= \sum_s \frac{m \, n_M}{\Bexm^2}(\partial \bE_\perp/\partial t) $. This polarization current can also be derived directly from the polarization drift of ions and electrons. The polarization drift velocity is given by $v_{\text{pol}}=1/(\omega_{cs}B)(\partial \bE_\perp/\partial t) = m_s/(q_sB^2)(\partial \bE_\perp/\partial t)$ which generates the corresponding current $\bJ_{\text{pol}}=en_{Mi}v_{\text{pol},i}-en_{Me}v_{\text{pol},e}=\sum_s \frac{m_s \, n_M}{\Bexm^2}(\partial \bE_\perp/\partial t)$ \cite{brizard2007foundations}. To our knowledge, the second term in Eq. \eqref{eq:Plin} is typically not considered in gyrokinetic simulations. We can estimate its contribution. The Faraday's law gives $B/E \sim k/\omega \sim 1/v_{\text{phase}}$, where $v_{\text{phase}}$ is the phase velocity of the wave. For the Alfv\'en waves, we have $v_{\text{phase}} \sim V_{A}$, and the ratio of the second term to the first term is typically $u_{\parallel}/V_{A} \ll 1$. Thus the second term is negligible under most conditions, but it may become significant in the presence of supersonic flows.

In the special case of a centered Maxwellian, we obtain
\begin{align}
	\mathbf{P} (t,\bx)&=  \frac{m \, n_M(\bx) \bE_\perp(t,\bx) }{|\Bex(\bx)|^2}, \label{eq:Plin-cent} \\
	\mathbf{M} (t,\bx) &=  - \bex \int \mu f \Bpstar  \dd \vpar \dd \mu . \label{eq:Mlin-cent}  
\end{align}
We shall consider only this case in the sequel of the paper.

We introduce the  Alfv\'en velocity of the particle species
\begin{equation}
    V_{A}(\bx) = \frac{|\Bex|}{\sqrt{\mu_0  m n_M}}
\end{equation}
to express the polarization of a given particle species Eq. \eqref{eq:Plin-cent} by
\begin{equation}\label{eq:polarizationAlfven}
	\boldsymbol{P} (t,\bx) = \frac{1}{\mu_0V_{A}(\bx)^2}\bE_\perp(t,\bx).
\end{equation}

Finally, we define the displacement field $\bD$ and the magnetic field intensity $\bH$ 
for several drift-kinetic species characterized by their polarization $\boldsymbol{P}$ given by Eq. \eqref{eq:polarizationAlfven} and magnetization $\boldsymbol{M}$ given by Eq. \eqref{eq:Mlin-cent}.
\begin{align}
	\bD &= \epsilon_0 \bE + \sum_s\boldsymbol{P}_s = \epsilon_0 \left(\bE + \sum_s \frac{c^2}{V_{A,s}^2}\bE_\perp  \right),\label{def:Dlin-cent}\\
	\bH &= \frac{1}{\mu_0}(\Bex+\bB) - \sum_s\boldsymbol{M}_s = 
	\frac{1}{\mu_0}(\Bex+\bB) + \sum_s\left( \int \mu f_s \Bpstar_s  \dd \vpar \dd \mu\right) \bex. \label{def:Hlin-cent}
\end{align}
We observe that, in the drift-kinetic model, the vacuum dielectric constant is not $\epsilon_0$ but rather a large tensorial quantity, $\epsilon_\perp(\boldsymbol{k})\epsilon_0$, that reflects the physics of the polarization drift intrinsically built in the definition of a gyrokinetic (GK) plasma. This corresponds to the dielectric constant of the "gyrokinetic vacuum" introduced in \cite{krommes1993dielectric}.

\begin{remark}
	In practical applications it is useful to estimate the magnitude of the different terms defining the $\bD$ and $\bH$. Let us consider a Tokamak with $B_{ext} = 3\,T$ and $n_{M,e}= 10^{20} \, m^{-3}$.
	Then with the physical parameters $m_e= 9.1095 \times 10^{-31} \, kg$, $\epsilon_0 = 8.85 \times 10^{-12} \, F/m$, $\mu_0 = 1.2566 \times 10^{-6} \, H / m$. 
	We have 
	$$\frac{ m_e n_{M,e}(\bx)}{|\Bex(\bx)|^2} = 1.01 \times 10^{-11} \, F/m, ~~~ V_{A,e} = 2.8 \times 10^8 m/s. $$
	This means that in a typical tokamak the electron polarization is of the same order of magnitude as $\epsilon_0$ and the electrons' Alfv\'en velocity is close to the speed of light. From a numerical point of view, the CFL is not reduced when neglecting $\epsilon_0 \bE$. \\ 
	For comparison, as $m_i = 1.6726 \times 10^{-27} \, kg$, we have
	$$\frac{ m_i n_{M,i}(\bx)}{|\Bex(\bx)|^2} = 1.85 \times 10^{-8} \, F/m, ~~~ V_{A,i} = 6.5 \times 10^6 m/s. $$
\end{remark}

On the other hand, with this linearization, the $K_2$ terms do not enter anymore in the particles' equations of motion, which we shall compute explicitly for this case. As $K_0+K_1$ does not depend on $\bE$ we have, decomposing $\bB$ into $\bB=\Bpar\bex + \Bperp$
\begin{equation}\label{eq:dkdX}
\fract{}{\bX}(K_0+K_1)= \mu (\nabla\Bexm +  \nabla(\bex\cdot \bB)) = \mu (\nabla\Bexm +  \nabla\Bpar).
\end{equation}
Moreover, it holds
\begin{equation}\label{eq:dkdvpar}
	\fracp{}{\vpar}(K_0+K_1)= m\vpar.
\end{equation}
So that the equations of motion \eqref{eq:Xdot-2}-\eqref{eq:Vdot-2} become
\begin{align}
	\fract{\bX}{t}&=\Vpar\frac{\boldsymboltar}{\Bpstar} 
	+ \frac{1}{\Bpstar} \left(\bE\times\bex -\frac {\mu}{q_e} \nabla \Bpartot\times\bex \right) = \vgc ,  \label{eq:Xdot}\\
	\fract{\Vpar}{t}&= \frac {q}{m}\frac{\boldsymboltar}{\Bpstar} \cdot \left( \bE - \frac {\mu}{q} \nabla \Bpartot\right)  = \agc .\label{eq:Vdot}
\end{align}
which is associated to the drift-kinetic Vlasov equation 
\begin{equation}\label{eq:VlasovDK}
    \fracp{f}{t} + \fract{\bX}{t}\cdot\nabla_x f +  \fract{\Vpar}{t}\fracp{f}{\Vpar}=0.
\end{equation}

\section{\label{sec:DKFK}The fully kinetic ions drift-kinetic electrons Vlasov--Maxwell system}

We assume one species of ions with distribution function $f_i$ solving the fully kinetic Vlasov equation
\begin{equation}
    \fracp{f_i}{t} + \bv\cdot\nabla_x f_i +  \frac{q_i}{m_i}(\bE+\bv\times\Btot)\cdot\fracp{f_i}{\bv}=0.
\end{equation}
The related equations of motion for the ions are 
\begin{align}
	\fract{\bX}{t}&=\bV, \label{eq:Xdoti}\\
	\fract{\bV}{t}&= \frac {q_i}{m_i}\left( \bE + \bV\times\Btot\right).\label{eq:Vdoti}
\end{align}

On the other hand, the distribution function $f_e$ of the electrons satisfies the drift-kinetic Vlasov Eq. \eqref{eq:VlasovDK}
\begin{equation}
    \fracp{f_e}{t} + \fract{\bX}{t}\cdot\nabla_x f_e +  \fract{\Vpar}{t}\fracp{f_e}{\Vpar}=0,
\end{equation}
where the characteristics are defined by Eqs. \eqref{eq:Xdot}-\eqref{eq:Vdot}.
The characteristics define the guiding center velocity $\vgc$ of the electrons, and 
the guiding center parallel acceleration of the electrons $\agc$.
We then define the total charge density
\begin{equation}\label{eq:rhotot}
    \rho = q_e \int  f_e \Bpstar \dd \vpar \dd \mu + q_i \int f_i \dd \bv = \rhogc + \rho_i,
    \end{equation}
and total current density
\begin{equation}\label{eq:jtot}
\bJ = q_e \int   \vgc  f_e \Bpstar \dd \vpar \dd \mu + q_i \int f_i \bv \dd \bv = \jgc +\bJ_i,
\end{equation}

The Amp\`ere law written in terms of $\bD$ and $\bH$ is then
\begin{equation}\label{eq:ampereDKFK0}
	\fracp{\bD}{t} -  \nabla\times\bH 
	= - \bJ = -\jgc - \bJ_i,
\end{equation}
with $\jgc = q_e  \int   \vgc  f_e  \Bpstar \dd \vpar \dd \mu$
and the  Gauss law becomes
\begin{equation}\label{eq:GaussDKFK0}
	\nabla\cdot\bD = \rho = \rhogc  + \rho_i
\end{equation}

In addition, the definition of the fields $\bE=-\fracp{\bA}{t}-\nabla\phi$ and $\bB=\nabla\times\bA$ yields Faraday's law $\fracp{\bB}{t}+\nabla\times\bE=0$ and $\nabla\cdot\bB =0$ so that we get the following Maxwell equations for our hybrid fully kinetic and drift-kinetic model
\begin{align}
\fracp{\bD}{t} - \nabla \times \bH &= -(\jgc +\bJ_i) , \label{eq:ampereDKFK}\\
\fracp{\bB}{t} + \nabla \times \bE  &= 0,\label{eq:faradayDKFK}\\
\nabla \cdot \bD &= \rhogc + \rho_i , \label{eq:GaussDKFK}\\
\nabla \cdot  \bB &= 0.
\end{align}
Equations \eqref{def:Dlin-cent}--\eqref{def:Hlin-cent} with only the electrons as a drift-kinetic species relate $\bD$ and $\bE$ on the one hand, and $\bH$ and $\bB$.

Finally, to complete our model,  the ions equations of motion are given by Eqs. \eqref{eq:Xdoti}--\eqref{eq:Vdoti} and the electrons equations of motion are given by Eqs. \eqref{eq:Xdot}--\eqref{eq:Vdot}.

\section{\label{sec:mimeticFD}Discretization with Mimetic Finite Differences}

This section introduces our geometric discretization of the model. We follow the method developed in \cite{Kormann2024A-Dual-Grid} for the fully kinetic Vlasov--Maxwell model. In the following subsection, we review the field discretization by Mimetic Finite Differences following the notation in \cite{Kormann2024A-Dual-Grid}. Then, we will the discretize the action of the drift-kinetic system and derive the discrete equations of motion by a discrete action principle.

\subsection{Mimetic Finite Differences}

The main idea is that the unknowns for the potentials, electric, and magnetic fields have different physical meanings, which results in a different discretization on a given grid: 
Potentials are naturally evaluated at points, the action of a force is measured through its circulation along a path, a  current is the flux through a surface of a current density,
a total charge in a volume is integral over volume of the charge density.
These kinds of physical quantifies are discretized accordingly in the Mimetic Finite Differences framework. More precisely potentials will be approximated by point values, corresponding to standard Finite Difference, whereas in a non standard manner, quantities related to a force, like the electric field $\bE$ have integrals over the edge of the mesh as a discrete representation, quantities related to a current like $\bJ$ and also $\bB$ are approximated by face integrals, and charge densities like $\rho$ are approximated by cell integrals on the mesh.
Figure \ref{fig:CellComplexes} shows the corresponding degrees of freedom. On the left, the vertices of the grid are used to discretize the potentials, then the edges of the mesh in each direction will be used for the three components of for example $\bE$, the fluxes through the faces orthogonal to each direction, will be used to discretize the three components of $\bB$ and densities will be discretized as integrals over each cell of the mesh.

These definitions of the approximate quantities allow one to write an exact discrete form of Faraday's equation in terms of $\bE$ and $\bB$. 
Indeed, using the Stokes theorem on each face $\mathbf{S}$ of the grid, whose boundary $\partial S$ consists of 4 edges, and on each cell $\mathbf{V}$ of the grid, whose boundary $\partial V$ consists of 6 faces, for $\nabla\cdot\bB=0$, we find:
\begin{align}
    \oint_{\partial \mathbf{S}}  \mathbf{E}\cdot \dd \mathbf{\ell} &=  \int_\mathbf{S} \left( -\fracp{ \mathbf{B}}{t} \right)\cdot \dd \mathbf{S}, \\
    \oint_{\partial \mathbf{V}} \mathbf{B}  \cdot \dd \mathbf{S} &= 0.
\end{align}
In order to apply Stokes' formula also for the Ampere and Poisson equations, we need to have the displacement field  $\bD$ discretized by fluxes through faces and the magnetic field $\bH$ by edge integrals. This corresponds to different set of unknowns that will defined on the dual grid whose vertices are the centers of the primal grid cells. The same procedure applied on the dual grid, with $\tilde{\mathbf{S}}$ and $\tilde{\mathbf{V}}$ now being the faces and cells of the dual grid, then yields the exact relations
\begin {align}
    \oint_{\partial \tilde{\mathbf{S}}}  \mathbf{H}\cdot \dd \tilde{\mathbf{\ell}} &=  \int_{\tilde{\mathbf{S}}} \left( \mathbf{J} + \fracp{\mathbf{D}}{t}\right) \cdot \dd \tilde{\mathbf{S}}, \\
    \oint_{\partial \tilde{\mathbf{\mathbf{V}}}} \mathbf{D}  \cdot \dd \tilde{\mathbf{S}} &= \int_{\tilde{\mathbf{V}}} \rho \dd \tilde{\mathbf{V}}. 
\end{align}

\begin{figure}
    \centering
    \includegraphics[width=0.5\linewidth]{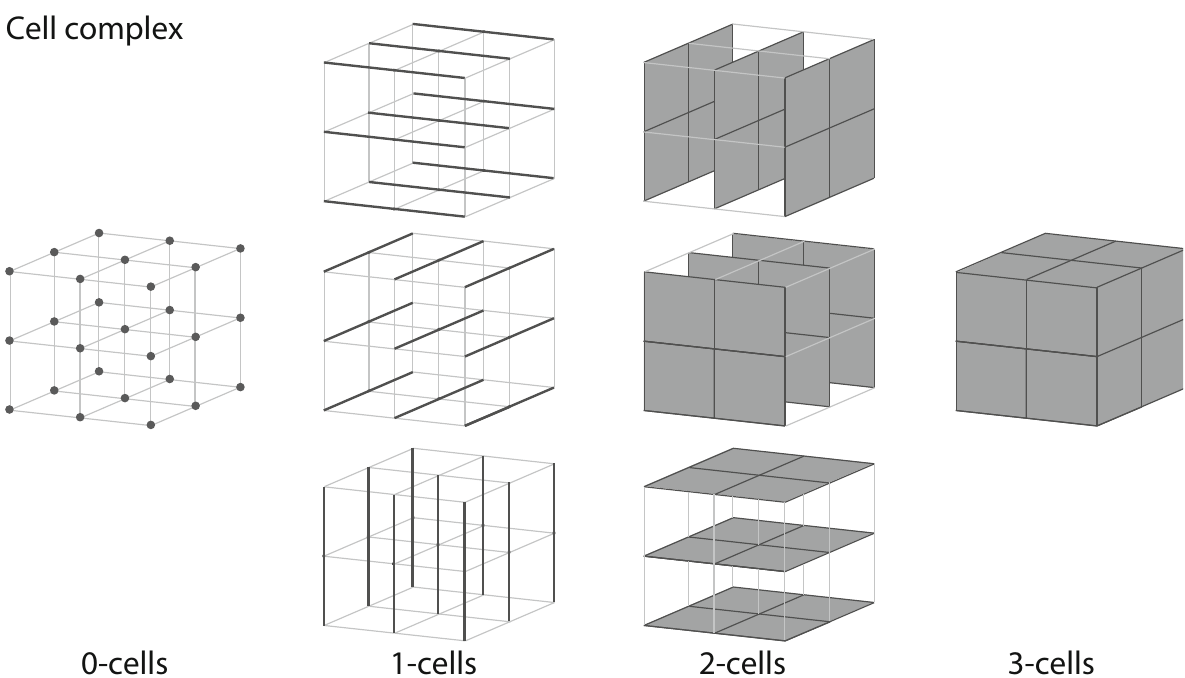}
    \caption{Discrete unknowns on a Cartesian grid}
    \label{fig:CellComplexes}
\end{figure}

We have now two discretizations of the electric field, namely $\bE$ and $\bD$, and of the magnetic field, namely $\bH$ and $\bB$, which are respectively edge integrals and fluxes through faces on the dual and primal grid. 
As one edge of the dual grid goes through exactly one face the primal grid and vice-versa, these unknowns are related to each other by a simple rescaling for a second order discretization. This corresponds to the Yee scheme.
Higher order scheme can also be achieved by a higher order interpolation and projection on the other mesh. The error in the scheme is only coming from this procedure as the Stokes equations are exact.
 
 Let us now formalize this construction:
 We consider two staggered tensor-product grids, a primal with vertices at given positions $(x_i,y_j,z_k)$, $i=1,\ldots,M_1$, $j=1,\ldots,M_2$, $k=1,\ldots,M_3$, where $M_{\ell}$, $\ell=1,2,3$, is the number of grid points in direction $\ell$, and a dual grid with vertices $(x_{i+1/2},y_{j+1/2},z_{k+1/2})$ placed at the center of each cell of the primal grid, i.e., we have e.g.~$x_{i+1/2} = \frac{x_{i+1}+x_i}{2}$.

We use a four-field discretization of Maxwell's equations, where the scalar and vector potentials as well as the electric and magnetic fields are discretized on the primal grid, and a dual grid is introduced for discretizing the $\bD$ and $\bH$ fields. On the primal grid, the scalar potential is defined as a node-based grid function, the vector potential and the electric field $\bE$ are defined as edge-based grid functions, and the magnetic field $\bB$ is defined as a face-based grid function. On the dual grid, $\bH$ is an edge-based grid function, $\bD$ is a face-based grid function as well as the current density $\bJ$, and the charge density $\rho$ is defined as a cell-based grid function. In order to define the degrees of freedom, let us introduce the so-called restriction operators on the primal grid:
\begin{itemize}
	\item $\cR_0$ associates to a scalar function its values at all the vertices of the primal grid:
	$ \cR_{0,(i,j,k)}(\phi) = \phi(x_i,y_j,z_k) =: \arrphi_{i,j,k}.$
	\item $\cR_1$ associates to a vector valued function the circulations on all the edges.  As there are three edges associated to a vertex we define
	$$\cR_{1,(i,j,k)} (\bE) = (\cR_{1,(i,j,k)} ^x(E_x), \cR_{1,(i,j,k)}^y (E_y), \cR_{1,(i,j,k)}^z (E_z))^\top,$$ with
	$\cR_{1,(i,j,k)}^x (E_x) = \int_{x_i}^{x_{i+1}} E_x(x,y_j,z_k) \dd x =: \arrE_{i+\half,j,k} =: \arrE^x_{i,j,k}$ the edge integral of $E_x$ along the $x$ direction, and similarly for the edges in the $y$ and $z$ directions. 
	\item $\cR_2$ associates to a vector valued function the fluxes through all faces.  As there are three faces associated to a vertex we define 
   $$\cR_{2,(i,j,k)} (\bB) = (\cR_{2,(i,j,k)}^x (B_x), \cR_{1,(i,j,k)}^y (B_y), \cR_{2,(i,j,k)}^z (B_z))^\top,$$ with
	$\cR_{2,(i,j,k)}^x (\bB) = \int_{y_j}^{y_{j+1}} \int_{z_{k}}^{z_{k+1}} B_x(x_i,y,z) \dd y \dd z =: \arrB_{i,j+\half,k+\half}=: \arrB^x_{i,j,k}$, the flux through the face orthogonal to the $x$ direction and similarly for the faces orthogonal to the $y$ and $z$ directions.
	\item $\cR_3$ associates to a scalar function its integrals on all the cells of the grid: 
	$\cR_{3,(i,j,k)}(\rho) = \int_{x_i}^{x_{i+1}}\int_{y_j}^{y_{j+1}}\int_{z_k}^{z_{k+1}} \rho(x,y,z) \dd x  \dd y \dd z = \arrrho_{i,j,k}.$
\end{itemize}
The restriction operators for the dual sequence are defined analogously using the vertices, edges, faces, and cells on the dual mesh:
\begin{itemize}
	\item $\tilde{\cR}_0$ associates to a scalar function its values at all the vertices of the dual grid:
	$ \tilde{\cR}_{0,(i,j,k)}(\phi) = \phi(x_{i+\half},y_{j+\half},z_{k+\half}) = \tilde{\arrphi}_{i+\half,j+\half,k+\half}.$
	\item $\tilde{\cR}_1$ associates to a vector valued function the circulations on all the edges of the dual grid. For the edges along the $x$ direction, we define
	$\tilde{\cR}_{1,(i,j,k)}^x (\bH) = \int_{x_{i-\half}}^{x_{i+\half}} H_x(x,y_{j+\half },z_{k+\half}) \dd x =: \tilde{\arrH}_{i,j+\half,k+\half} =: \tilde{\arrH}^x_{i,j,k}$, and similarly for the edges in the $y$ and $z$ directions. 
	\item $\tilde{\cR}_2$ associates to a vector valued function the integrals on all the fluxes through all the faces of the dual grid. For the faces orthogonal to the $x$ direction, we define
	$$\tilde{\cR}_{2,(i,j,k)}^x (\bD) = \int_{y_{j-\half}}^{y_{j+\half}} \int_{z_{k-\half}}^{z_{k+\half}} D_x(x_{i+\half},y,z) \dd z \dd y =: \tilde{\arrD}_{i+\half,j,k} =: \tilde{\arrD}^x_{i,j,k},$$ and similarly for the faces orthogonal to the $y$ and $z$ directions.
	\item $\tilde{\cR}_3$ associates to a scalar function its integrals on all the cells of the dual grid: 
	$\tilde{\cR}_{3,(i,j,k)}(\rho) = \int_{x_{i-\half}}^{x_{i+\half}}\int_{y_{j-\half}}^{y_{j+\half}}\int_{z_{k-\half}}^{z_{k+\half}} \rho(x,y,z) \dd x  \dd y \dd z = \tilde{\arrrho}_{i,j,k}.$
\end{itemize}
We collect the degrees of freedom into a vector in the following form $\arrB = (\arrB^x, \,\arrB^y,\, \arrB^z)^\top$, where the vectors $\arrB^{\alpha}$, $\alpha =x,y,z$ have the elements $\arrB^{\alpha}_{I} = \arrB^{\alpha}_{i,j,k}$ with $I = i + M_1 j + M_1 M_2 k$ and $\arrB^{\alpha}_{i,j,k}$ is as defined above as a component of the restriction operator. The other field vectors $\arrE$, $\tilde{\arrD}$, $\tilde{\arrH}$ are defined correspondingly.
Since we use a four field formulation, we need to map between the operators on the primal to the dual grid and vice versa. In the time stepping algorithm we will use to propagate $\tilde{\arrD}$ and $\arrB$. The mapping operators are called Hodge operators and they can be defined to arbitrary order of accuracy as was  described in \cite{Kormann2024A-Dual-Grid}. However, in this work we will consider only second-order Hodge operators, which then only amount to rescaling the unknowns from edge to face or point to volume. The resulting scheme is then analogous to the classical Yee scheme.

In practice, for mapping respectively a primal edge unknown $\arrF$ to a dual face unknown $\tilde{\arrF}$, a primal face unknown $\arrG$  to a dual edge unknown $\tilde{\arrG}$, and a primal cell unknown $\arrpsi$ to a dual vertex unknown $\tilde{\arrpsi}$, this yields:
\begin{align}\label{eq:hodge_operators}
	\tilde{\arrF}^x &= \frac{\Delta y \Delta z}{\Delta x}\arrF^x, \quad \tilde{\arrF}^y = \frac{\Delta x \Delta z}{\Delta y}\arrF^y, \quad \tilde{\arrF}^z = \frac{\Delta x \Delta y}{\Delta z}\arrF^z, \\
	\tilde{\arrG}^x &= \frac{\Delta x}{\Delta y \Delta z}\arrG^x, \quad \tilde{\arrG}^y = \frac{\Delta y}{\Delta x \Delta z}\arrG^y, \quad \tilde{\arrG}^z = \frac{\Delta z}{\Delta x \Delta y}\arrG^z, \\
	\tilde{\arrpsi} &= \frac{1}{\Delta x \Delta y \Delta z}\arrpsi \;.
\end{align}
From these expressions, we define the diagonal Hodge operators, $\matH_1$, $\matH_2$, and $\matH_3$, stacking the components for the vector fields
\begin{equation}\label{eq:diagonal_hodge_operators}
	\tilde{\arrF} = \matH_1 \arrF, ~~ \tilde{\arrG} = \matH_2 \arrG, ~~ \tilde{\arrpsi} = \matH_3 \arrpsi.
\end{equation}

Accordingly, we shall also approximate our integral restriction operators by the second order midpoint quadrature rule, so that
\begin{align}\label{eq:midpoint_rule}
	\cR_{1,(i,j,k)}^x (F_x) &\approx \Delta x \,F_x(x_{i+\half},y_j,z_k), \\
	\cR_{2,(i,j,k)}^x (G_x) &\approx \Delta y \Delta z \, G_x(x_i,y_{j+\half},z_{k+\half}), \\
	\cR_{3,(i,j,k)}^x (\psi) &\approx \Delta x \Delta y \Delta z \, \psi(x_{i+\half},y_{j+\half},z_{k+\half}) ,
\end{align}
and similarly for the other components and dual grid restrictions.

We observe that on a Cartesian grid with periodic boundary conditions a cell on the dual grid exactly matches a vertex on the primal grid and an edge on the primal grid exactly matches a face on the dual grid, and vice-versa. 
So with these notations for the reduction operators the indices  perfectly match, and we can define the following discrete scalar products that approximate the $L_2$ inner products
\begin{align*}
	\arrpsi\cdot\tilde{\arrpsi} &= \sum_{i,j,k} \arrpsi_{i,j,k} \tilde{\arrpsi}_{i,j,k}, \\
	\arrF\cdot\tilde{\arrF} &= \sum_{i,j,k} \arrF_{i+\half,j,k}\tilde{\arrF}_{i+\half,j,k}
	+ \arrF_{i,j+\half,k}\tilde{\arrF}_{i,j+\half,k} + \arrF_{i,j,k+\half}\tilde{\arrF}_{i,j,k+\half}\\
	&= \arrF^x\cdot\tilde{\arrF}^x + \arrF^y\cdot\tilde{\arrF}^y +\arrF^z\cdot\tilde{\arrF}^z\\
	\arrG\cdot\tilde{\arrG} &= \sum_{i,j,k} \arrG_{i,j+\half,k+\half}\tilde{\arrG}_{i,j+\half,k+\half}
	+ \arrG_{i+\half,j,k+\half}\tilde{\arrG}_{i+\half,j,k+\half} + \arrG_{i+\half,j+\half,k}\tilde{\arrG}_{i+\half,j+\half,k}\\
	&= \arrG^x\cdot\tilde{\arrG}^x + \arrG^y\cdot\tilde{\arrG}^y +\arrG^z\cdot\tilde{\arrG}^z \\
	\arrpsi\cdot\tilde{\arrpsi} &= \sum_{i,j,k} \arrpsi_{i+\half,j+\half,k+\half} \tilde{\arrpsi}_{i+\half,j+\half,k+\half}. 
\end{align*}  

Due to the definition of the degrees of freedom based on the restriction operators the discrete gradient, curl, and divergence operators have a simple form: Let us illustrate this on the example of the $x$-component of a scalar function $\phi$:
\begin{align*}
\cR^x_{1,(i,j,k)}\left(\frac{\dd\phi}{\dd x} \right) &=\int_{x_i}^{x_{i+1}} \frac{\dd\phi}{\dd x} (x,y_j,z_k)\dd x = \phi(x_{i+1},y_j,z_k) - \phi(x_{i},y_j,z_k) 
\\
&= (\matd_{M_1} \arrphi_{:,j,k})_{i} = ( (\matI_{M_3} \otimes \matI_{M_2} \otimes \matd_{M_1}) \arrphi)_{i,j,k}, 
\end{align*}
where $\matI_{M_{\ell}}$ denotes the identity matrix of size $M_{\ell} \times M_{\ell}$, $\ell=1,2,3$, $\otimes$ the Kronecker product of matrices and we define the one-dimensional discrete derivative operator as 
\begin{equation*}
	\mathbb{d}_{M_1} = \begin{pmatrix}
   -1 & 1 & 0 & \ldots & 0 \\
   0 & -1 & 1 & 0  & \\
   \vdots & & \ddots & \ddots & \\
   0 & & & -1 & 1 \\
   1 & 0 & \ldots & 0 & -1
   \end{pmatrix} \in \mathbb{R}^{M_1 \times M_1}.
\end{equation*} 
An analogous calculation on the dual grid yields $\tilde{\cR}^x_{1,(i,j,k)}\left(\frac{\dd\phi}{\dd x} \right) =(\tilde{\mathbb{d}}_{M_1}\tilde{\arrphi})_{i,j,k}$ with the adjoint derivative operator $\tilde{\matd}_{M_1} = - \matd_{M_1}^\top$. Performing similar calculations for all components of $\cR_1$ and for the curl and divergence components, the discrete gradient, curl, and divergence operators on the primal grid are defined as blocks of Kronecker products of the one dimensional derivative and identity matrices: 
\begin{equation} \label{hD0}
\matG = \begin{pmatrix}
\matI_{M_3} \otimes \matI_{M_2} \otimes \matd_{M_1} \\
\matI_{M_3} \otimes \matd_{M_2} \otimes \matI_{M_1} \\
\matd_{M_3} \otimes \matI_{M_2} \otimes \matI_{M_1}
\end{pmatrix},
\end{equation}
\begin{equation}\label{hD12}
\matC = \begin{pmatrix}
\matO_{M} & - \matd_{M_3} \otimes \matI_{M_2} \otimes \matI_{M_1}  & \matI_{M_3} \otimes \matd_{M_2} \otimes \matI_{M_3}\\
\matd_{M_3} \otimes \matI_{M_2} \otimes \matI_{M_1}  & \matO_{M} &  -\matI_{M_3} \otimes \matI_{M_2} \otimes \matd_{M_1}\\
-\matI_{M_3} \otimes \matd_{M_2} \otimes \matI_{M_1}  & \matI_{M_3} \otimes \matI_{M_2} \otimes \matd_{M_1} & \matO_{M} \\
\end{pmatrix}
\,  \text{ and } \,
\matD =  \matG^\top,
\end{equation}
where $\matO_M$ denotes the zero matrix of size $M \times M$ with $M=M_1\cdot M_2 \cdot M_3$. Recall that the degrees of freedom are ordered that the first block of $M$ entries belongs to the $x$-components on all point such that the first column of the curl matrix operates on the $x$-components of the degrees of freedom and the first row corresponds to the first component of the curl operation.
Their adjoint operators used on the dual grid are given as:
\begin{equation}\label{eq:relation_dual_derivative_ops}
\matG^\top = - \tilde{\matD}, \quad \matC^\top = \tilde{\matC}, \quad \matD^\top = - \tilde{\matG}.
\end{equation}

\subsection{The discrete action principle}

In this section, we derive a semi-discrete action principle for the hybrid model, which can contain both fully kinetic and drift-kinetic species. We will follow the derivation in \cite{Kormann2024A-Dual-Grid} where the discrete action principle for the pure fully kinetic case was derived. The fields are discretized based on Mimetic Finite Differences as explained in the previous subsection and use both fully kinetic ($\bX$, $\bV$) and drift-kinetic markers ($\bbX$, $\Vpar$) to represent the distribution function.
We will denote the positions of the drift-kinetic markers by $\bbX_p$ and their velocities by $\Vpar_p$.
Moreover, to each species of drift-kinetic markers, we will associate the corresponding linearized polarization and magnetization fields as defined in Eqs. \eqref{eq:polarizationAlfven} and \eqref{eq:Mlin-cent}.

For the field particle coupling we will use tensor product cardinal B-splines of arbitrary degree $d$ in each direction. The one dimensional B-splines are assumed to be even and integrate to 1. Their support consists of $d+1$ cells. We will denote the spline centered at the marker position by $S_p(\bx)=S(\bx-\bX_p)$.  

These discretizations are inserted into the drift-kinetic Lagrangian Eq. \eqref{eq:action} and the Lagrangian for the fully kinetic model.
The discrete Lagrangian will then consist of four parts
\begin{equation}\label{eq:Lh}
	\mathcal{L} = \mathcal{L}_{\text{kinetic}} + \mathcal{L}_{\text{dk}} + \mathcal{L}_{\text{pol}} + \mathcal{L}_{\text{fields}}
\end{equation}
with
\begin{equation}\label{eq:kinetic}
	\mathcal{L}_{\text{kinetic}} = \sum_{p=1}^{N_p} w_p (m_p \bV_p + q_p (\arr{\Aex}+ \arr{A})) \cdot \tilde{\mathcal{R}}_2  \left( \dot{\bX}_p S_p(\bx)\right) - (\frac 12 m_p \bV_p^2 - q_p \arr{\phi} \cdot \tilde{\mathcal{R}}_3(S_p(\bx))).
\end{equation}
\begin{multline}\label{eq:dk}
	\mathcal{L}_{\text{dk}} = \sum_{p=1}^{\bar{N}_p} w_p
	  \left( m_p \Vpar_p \bex(\bbX_p) \cdot \dot{\bbX}_p  +  q_p (\arr{A}_0 + \arr{A}) \cdot \tilde{\mathcal{R}}_2  \left( \dot{\bbX}_p S_p(\bx)\right)   - q_p\arr{\phi} \cdot \tilde{\mathcal{R}}_3(S_p(\bx)) \right.  \\
	  \left. - \frac{m_p}{2} \Vpar_p^2 - \mu_p \Bexm -\mu_p (\matC\arrA) \cdot \tilde{\mathcal{R}}_1 \left(S_p(\bx)\bex(\bbX_p)\right)\right),
\end{multline}
where $\arr{A}$ and $\arr{\phi}$ represent the arrays of degrees of freedom, $w_p$ is the particle weight.	  
The underlying definition of the electric and magnetic field degrees of freedom are
\begin{equation}\label{eq:eb_def}
	\arr{E} = - \dot{\arr{A}} - \mathbb{G} \arr{\phi}, \quad \arr{B} = \mathbb{C} \arr{A}.
\end{equation}
Moreover,
\begin{equation}
	\mathcal{L}_{\text{fields}} = 	\half \tilde{\arr{D}} \cdot \arr{E} - \half \arr{B}  \cdot  \tilde{\arr{H}} 
\end{equation}
with, using the expressions of polarization and magnetization from Eqs. \eqref{eq:polarizationAlfven} and \eqref{eq:Mlin-cent},
\begin{align}
		\tilde{\arrD} &= \tilde{\cR}_2(\epsilon_0 \bE^R + \boldsymbol{P}^R) = \epsilon_0 \left(\matH_1\arrE + \sum_s
		\tilde{\cR}_2\left(\frac{c^2}{V_{A,s}^2}\bE_\perp^R \right)  \right), \label{eq:polarizationHodge}\\
		\tilde{\arrH} &= \tilde{\cR}_1\left(\frac{1}{\mu_0}(\Bex+\bB^R) - \boldsymbol{M}\right) = \frac{1}{\mu_0}\left(\matH_2(\arrB_{\mathrm{ext}}+\arrB) +
		\sum_{p=1}^{\bar{N}_p} w_p \mu_p \Bpstarp\tilde{\cR}_1\left(S(\bx-\bbX_p) \bex\right)\right)
\end{align}
where the restriction operators are evaluated with the midpoint rule following Eq. \eqref{eq:midpoint_rule}.
In order to apply the restriction operators $\tilde{\cR}_1$ and $\tilde{\cR}_2$, the fields, that we denote by $\bE^R$ and $\bB^R$, are reconstructed from their degrees of freedom in $\arrE$ and $\arrB$ by interpolation and histopolation as explained in details in \cite{Kormann2024A-Dual-Grid}.
The sum for the polarization term is performed only over the drift-kinetic species.

Our time stepping algorithm updates $\tilde{\arrD}$ and $\arrB$. This means in particular that to update $\arrB$ in the Faraday equation which needs $\arrE$,
we will need to compute $\arrE$ from $\tilde{\arrD}$, which means that we need to invert Eq. \eqref{eq:polarizationHodge}. This is straightforward in general.
However, when $\epsilon_0|\bE| \ll |\boldsymbol{P}|$, then only the perpendicular component of the electric field can be determined in this way.

\subsection{The discrete dynamical system}

From this action principle, the equations of motions can be obtained by performing the variations with respect to all the dynamical variables, which yield the corresponding Euler--Lagrange equations $\fract{}{t}\fracp{\mathcal{L}}{\dot{q}} = \nabla_q \mathcal{L}$.

First for the kinetic particles, the Euler--Lagrange equations are computed like in \cite{Kormann2024A-Dual-Grid}:
\begin{align}
	\fract{\bX_p}{t} &= \bV_p \label{eq:varX}\\
	\fract{\bV_p}{t} &= \frac{q_p}{m_p} (\bE^S(\bX_p) + \bV_p\times \bB^S(\bX_p) ).\label{eq:varV}
\end{align}
where $\bE^S(\bX_p) = -\fracp{\bA^S}{t}(\bX_p)-\nabla \phi^S(\bX_p)$ and $\bB^S(\bX_p) = \nabla \times \bA^S(\bX_p)$, and 
\begin{equation*}
	\bA^S(\bX_p) = \begin{pmatrix}
		\arrA^x \cdot \tilde{\cR}_{2}^x(S(\bx-\bX_p)) \\
		\arrA^y \cdot\tilde{\cR}_{2}^y(S(\bx-\bX_p)) \\
		\arrA^z \cdot\tilde{\cR}_{2}^z(S(\bx-\bX_p))
	\end{pmatrix}, \quad 
	\phi^S(\bX_p) = \arrphi  \cdot\tilde{\cR}_{3}(S(\bx-\bX_p)). 
\end{equation*}
Let us notice that the smoothed fields satisfy the important properties 
\begin{align}
	\nabla\phi^S(\bX_p) &= (\matG\arrphi)^x \cdot\tilde{\cR}_{2}^x( S(\bx-\bX_p))
+ (\matG\arrphi)^y \cdot\tilde{\cR}_{2}^y( S(\bx-\bX_p)) + (\matG\arrphi)^z \cdot\tilde{\cR}_{2}^z( S(\bx-\bX_p)), \label{eq:gradphi_by_parts}\\
	\nabla\times \bA^S(\bX_p) &= (\matC\arrA)^x\cdot\tilde{\cR}_{1}^x( S(\bx-\bX_p)) + (\matC\arrA)^y\cdot\tilde{\cR}_{1}^y( S(\bx-\bX_p)) + (\matC\arrA)^z\cdot\tilde{\cR}_{1}^z( S(\bx-\bX_p)),\\
	\nabla\cdot \bB^S(\bX_p) &= \matD\arrB \cdot \tilde{\cR}_{0}(S(\bx-\bX_p)).
\end{align}
These equalities follow from the following computation, that can be done in 1D as we have a tensor product structure 
\begin{multline}
	\fracp{\phi^S}{x}(X_p) = -\sum_i \arrphi_i\int_{x_{i-\frac 12}}^{x_{i+\frac 12}} S'(x-X_p) \dd x  = 
	-\sum_i \arrphi_i (S(x_{i+\frac 12}-X_p)- S(x_{i-\frac 12}-X_p))\\
	= \sum_i(\arrphi_{i+1}-\arrphi_i){S}(x_{i+\frac 12}-X_p)
\end{multline}
by rearranging the sum and assuming periodic boundary conditions. We get the desired result by computing all the needed derivatives in the same way.

We now compute the Euler--Lagrange equations for the guiding centers.
First for $\bbX_p$ we have
\begin{equation}\label{eq:varXgc}
	 m_p \fract{\Vpar_p}{t}\bexp + q_p \boldsymboltarp \times \fract{\bbX_p}{t} = q_p \bE^S(\bbX_p) -\mu_p \nabla\Bpartotp,
\end{equation}
where we denote by $\Bpartotp=\Bexm(\bbX_p) + \bexp\cdot\bB^S(\bbX_p)$.

Now, the Euler--Lagrange equation for $\Vpar_p$ yields
\begin{equation}\label{eq:varVpar}
	\bexp \cdot \fract{\bbX_p}{t} = \Vpar_p.
\end{equation}
In order to decouple $\fract{\Vpar_p}{t}$ and $\fract{\bbX_p}{t}$, we first take the cross product of Eq. \eqref{eq:varXgc} with $\bexp$ to eliminate the first term. This yields
\begin{align}
	\left(\bE^S(\bbX_p) -\frac{\mu_p}{q_p} \nabla\Bpartotp\right)\times \bexp &=(\boldsymboltarp \times \fract{\bbX_p}{t})\times \bexp  \\
	&= (\bexp \cdot\boldsymboltarp) \fract{\bbX_p}{t} - (\bexp \cdot\fract{\bbX_p}{t}) \boldsymboltarp \\
	&= \Bpstarp \fract{\bbX_p}{t} - \Vpar_p \boldsymboltarp
\end{align} 
where we used an algebraic identity in the second expression, and $\Bpstar= \bex \cdot\boldsymboltar$ as well as Eq. \eqref{eq:varVpar} in the last expression. 
On the other hand, taking the dot product of Eq. \eqref{eq:varXgc} with $\boldsymboltarp$ yields
\begin{equation}
	\Bpstarp \fract{\Vpar_p}{t} = \boldsymboltarp \cdot \left(\frac{q_p}{m_p} \bE^S(\bbX_p) -\frac{\mu_p}{m_p} \nabla\Bpartotp\right).
\end{equation}
Summing up, we get the following equations of motion for the guiding centers:
\begin{align}
	\fract{\bbX_p}{t} &= \Vpar_p \frac{\boldsymboltarp}{\Bpstarp}  + \frac{1}{\Bpstarp} \left(\bE^S(\bbX_p) -\frac{\mu_p}{q_p} \nabla\Bpartotp\right)\times \bexp =: \bbV_p \\
	\fract{\Vpar_p}{t} &= \frac{\boldsymboltarp}{\Bpstarp} \cdot \left(\frac{q_p}{m_p} \bE^S(\bbX_p) -\frac{\mu_p}{m_p} \nabla\Bpartotp\right).
\end{align} 

The Euler--Lagrange equation for $\arr{A}$ yields
\begin{equation}\label{eq:AmpereDisc}
	\fract{\tilde{\arr{D}}}{t} - \mathbb{C}^\top \tilde{\arr{H}} = 
	-\sum_{p=1}^{N_p} w_p q_p \tilde{\mathcal{R}}_2  \left( \bV_p S(\bx-\bX_p)\right)
	-\sum_{p=1}^{\bar{N}_p} w_p q_p \tilde{\mathcal{R}}_2  \left(\bbV_p S(\bx-\bbX_p)\right) =:  -\arr{J} -\arr{J}_{gc}
\end{equation}
and the Euler--Lagrange equation for $\arr{\phi}$ yields
\begin{equation}\label{eq:PoissonDisc}
	\mathbb{G}^\top\tilde{\arr{D}} = \sum_{p=1}^{N_p} w_p q_p \tilde{\mathcal{R}}_3  \left( S(\bx-\bX_p)\right) + \sum_{p=1}^{\bar{N}_p} w_p q_p \tilde{\mathcal{R}}_3  \left( S(\bx-\bbX_p)\right) =: \arr{\rho} + \arr{\rho}_{gc},
\end{equation}
where we denote by $\arr{\rho}$ and $\arrJ$ the charge and current density of the kinetic particles, and by $\arr{\rho}_{gc}$ and $\arrJ_{gc}$ the charge and current density of the drift-kinetic particles.

On the other hand, we get from Eq. \eqref{eq:eb_def} the Faraday equation
\begin{equation}\label{eq:FaradayDisc}
	\fract{\arr{B}}{t} + \mathbb{C}\arr{E} = 0
\end{equation}
as well as $\mathbb{D}\arr{B}=0$.

\subsection{Semi-discrete energy conservation}

As the Lagrangian does not depend explicitly on time, the following discrete energy is conserved:
\begin{equation}
	\mathcal{H}_h = \half \tilde{\arr{D}} \cdot \arr{E} + \half \tilde{\arr{H}} \cdot \arr{B} 
	+ \sum_{p=1}^{N_p} w_p
		\frac{m_p}{2} \bV_p^2
	+ \sum_{p=1}^{\bar{N}_p} w_p\left( 
		\frac{m_p}{2} \Vpar_p^2 + \mu_p \Bpartotp
	\right) 
\end{equation}
It is instructive to verify the conservation by explicit computation. First we add the dot product of the Ampere equation \eqref{eq:AmpereDisc} with $\arr{E}$ to the dot product of the Faraday equation \eqref{eq:FaradayDisc} with $\tilde{\arr{H}}$, we find
\begin{equation}
	\fract{}{t}\left(\half \tilde{\arr{D}} \cdot \arr{E} + \half \tilde{\arr{H}} \cdot \arr{B}\right) = - (\arr{J} + \arr{J}_{gc})\cdot \arr{E} = 
	- \sum_{p=1}^{N_p} w_p q_p \bE^S(\bX_p)\cdot \bV_p
	- \sum_{p=1}^{\bar{N}_p} w_p q_p \bE^S(\bbX_p)\cdot \fract{\bbX_p}{t},
\end{equation}
by definition of $\bE^S$. 

Then taking the dot product of  Eq. \eqref{eq:varV} with $m_p \bV_p$ we find
\begin{equation}\label{eq:dtV}	
	\sum_{p=1}^{N_p} w_p m_p \bV_p \cdot \fract{\bV_p}{t} = \sum_{p=1}^{N_p} w_p  q_p \bV_p  \cdot  \bE^S(\bX_p) = \arrJ \cdot \arr{E}.
\end{equation}

We now take the dot product of Eq. \eqref{eq:varXgc} with $q_p \bE^S(\bbX_p)$:
\begin{equation} \label{eq:EdotX}
	q_p \bE^S(\bbX_p) \cdot \fract{\bbX_p}{t} = \Vpar_p \frac{\boldsymboltarp}{\Bpstarp} \cdot q_p \bE^S(\bbX_p) - \frac{\mu_p}{\Bpstarp}  (\nabla\Bpartotp\times \bexp) \cdot \bE^S(\bbX_p) .
\end{equation}
On the other hand multiplying Eq. \eqref{eq:varVpar} by $m_p \Vpar_p$ we find
\begin{equation}\label{eq:dtVpar}
	\half m_p\fract{\Vpar_p}{t} = \Vpar_p \frac{\boldsymboltarp}{\Bpstarp} \cdot \left(q_p \bE^S(\bbX_p) -\mu_p \nabla\Bpartotp\right).
\end{equation}
We also have
\begin{multline}\label{eq:dtmu}
	\fract{\mu_p \Bpartotp}{t} = \mu_p \fract{\bbX_p}{t}\cdot\nabla \Bpartotp \\
	= \Vpar_p \mu_p\frac{\boldsymboltarp}{\Bpstarp}  \cdot\nabla \Bpartotp + \frac{\mu_p}{\Bpstarp} \nabla\Bpartotp\cdot(\bE^S(\bbX_p) \times\bexp).
\end{multline}
Then adding Eqs. \eqref{eq:dtmu} and \eqref{eq:dtVpar} and subtracting Eq. \eqref{eq:EdotX} yields 0. Then summing the particle contributions and multiplying by $w_p$, we can add them to the field's contribution to get the energy conservation.

\begin{remark}
In this computation, we observe in particular that the cancellation needs the smoothed electric field to be consistent with the definition of the current, so that
\begin{equation}
	\arr{J} \cdot \arr{E} = \sum_{p=1}^{N_p} w_p q_p \bE^S(\bX_p)\cdot \bV_p, \text{ and }
	\arr{J}_{gc}\cdot \arr{E} = \sum_{p=1}^{\bar{N}_p} w_p q_p \bE^S(\bbX_p)\cdot \fract{\bbX_p}{t}.
\end{equation} 
The evaluation of the other field quantities at the particle positions can use any approximation without affecting energy conservation.
\end{remark}

\section{\label{sec:timeDisc}Low-storage Runge--Kutta scheme for time discretization}
In time, we discretize the equations based on an explicit Runge--Kutta scheme. Since the amount of data to store the particles in 4D and 6D phase-space is very high, we use a low-storage variant that is optimized to save storage. We use the Williamson (2N) methods \cite{williamson1980lsrk}, which are low-storage Runge--Kutta (LSRK) schemes for solving ordinary differential equations of the form
$$
u' = F(u(t)), \quad u(0) = u_0,
$$
using an $s$-stage approach with minimum memory requirements of only one additional copy. Williamson's method is defined as follows:
\begin{equation}
\begin{aligned}
    & \quad S_1 := u^n \\
    & \quad \text{for } i = 1:s \ \text{do} \\
    & \quad \quad S_2 := A_i S_2 + \Delta t F(S_1) \\
    & \quad \quad S_1 := S_1 + B_i S_2 \\
    & \quad \text{end} \\
    & \quad u^{n+1} = S_1
\end{aligned}    
\end{equation}
The coefficients are given as \( A_1, \dots, A_s \) and \( B_1, \dots, B_s \), where \( A_1 = 0 \). $S_2$ is initialized at the first pass by setting \( A_1\) to zero.
 All second-order, many third-order, and a few fourth-order methods can be can be cast in 2N-storage
format \cite{williamson1980lsrk}. Below, we list several sets of coefficients \( A \) and \( B \):
\begin{itemize}
    \item \textbf{Explicit Euler (1-Stage method):}
    \[
    A = \{0.0\}; \quad B = \{1.0\}.
    \]

    \item \textbf{2-Stage Methods:}
    \begin{itemize}
        \item \textbf{Heun Method:}
        \[
        A = \{0.0, -1.0\}; \quad B = \{1.0, 0.5\}.
        \]
        \item \textbf{Ralston Method:}
        \[
        A = \{0.0, -5/9\}; \quad B = \{2/3, 3/4\}.
        \]
    \end{itemize}

    \item \textbf{3-Stage Method (from Williamson 1980 \cite{williamson1980lsrk}):}
        \[
        A = \{0.0, -5/9, -153/128\}; \quad B = \{1/3, 15/16, 8/15\}.
        \]
    \item \textbf{5-Stage Method (from Carpenter \& Kennedy 1994 \cite{carpenter1994fourth}):}
    \[
    \begin{aligned}
        A & = \{0.0, -0.417890474499852, -1.19215169464268, -1.69778469247153, -1.51418344425716
\};\\  
        B &= \{0.149659021999229, 0.379210312999627, 0.822955029386982,
           0.699450455949122, 0.153057247968152
\}.
    \end{aligned}
    \]
\end{itemize}
The 3-stage method described above is a third-order method, while the 5-stage method is a fourth-order method.

The process for each Runge-Kutta (RK) stage in the Vlasov-Maxwell system consists of sequential steps to update the electromagnetic fields and particle quantities. For each RK stage, the following steps are performed sequentially:
\begin{enumerate}
    \item \textbf{Ampere's Law}: Update the electric displacement field $\tilde{\arrD}_{\text{new}}$ using $\arrH_{\text{old}}, \; \arrJ_{\text{old}}$, $\arrD_{\text{old}}$

    \item \textbf{Push Particles and Deposit Current}:
    \begin{itemize}
        \item Initialize the current density: Set $\arrJ = 0$
        \item Push particles: Update particle positions and velocities ($\arrX_{\text{new}}, \arrV_{\text{new}}$) using $\arrE_{\text{old}}$ and $\boldsymbol{B}_{\text{old}}$
        \item Deposit the current: Deposit the current $\arrV_{\text{new}}$ based on particle updates ($\arrX_{\text{new}}, \arrV_{\text{new}}$)
        \item Synchronize current through a post-particle loop synchronization
    \end{itemize}

    \item \textbf{Faraday's Law}: Update the magnetic field $\tilde{\arrB}_{\text{new}}$ using $\arrE_{\text{old}}$, $\tilde{\arrB}_{\text{old}}$

    \item \textbf{Hodge for $\arrB$ and $\arrH$}: Update the magnetic field intensity $\arrH_{\text{new}}$ using $\arrB_{\text{new}}$

    \item \textbf{Hodge for $\tilde{\arrD}$ and $\arrE$}: Update the electric field $\arrE_{\text{new}}$ using  $\tilde{\arrD}_{\text{new}}$

\end{enumerate}
In Step 2, different models can be chosen to calculate particle motion and current. If fully kinetic models are used for all species, in Steps 4 and 5, the Hodge operators are constant scalings in each component on equidistant grids. If a drift-kinetic model is used for electrons and a fully kinetic model for ions, as noted in Eqs. \eqref{def:Dlin-cent} and \eqref{def:Hlin-cent}, the polarization $\arrP$ and magnetization $\arrM$ are calculated only from the electron species, and we have $\bm D = \epsilon_0 \left(\bm E + \frac{c^2}{V_{A,e}^2}\bm E_\perp  \right)$ for Eq.~\eqref{def:Dlin-cent}. If drift-kinetic models are used for all species, the relation $1/V_A^2 = \sum_s 1/V_{A,s}^2$ can be applied to simplify the calculation. In the following sections, for the initial implementations in a uniform plasma we consider only $\mu=0$ and the magnetization $\arrM=0$.

\section{\label{sec:dispersion}Dispersion relation for the drift-kinetic model}
To compute the dispersion relation, we simplify the model for the case of a slab with a constant and uniform magnetic field $B_{ext}$ in the $z$ direction
and consider only $\mu=0$.
In this case  $\boldsymboltar=\bB +\Bex$ and $\Bpstar=(\bB +\Bex)\cdot \bex = B_z + B_{ext}$ do not depend on $\vpar$ and can be removed from the velocity integral when computing the charge and parallel current densities.
Decomposing $\bX=(X,Y,Z)$ into its components, and writing $\Vpar=V_z$, the characteristics Eqs. \eqref{eq:Xdot}-\eqref{eq:Vdot} become

\begin{align}
	\fract{X}{t}&=\frac{1}{\Bpstar}  \left(V_z B_x + E_y \right),  \label{eq:Xdot_slab}\\
	\fract{Y}{t}&=\frac{1}{\Bpstar}  \left(V_z B_y - E_x\right),  \label{eq:Ydot_slab}\\
	\fract{Z}{t}&=V_z,  \label{eq:Zdot_slab}\\
	\fract{V_z}{t}&= \frac {q_e}{m_e}\frac{\bB\cdot\bE+ B_{ext} E_z}{B_{ext}+B_z}.\label{eq:Vdot_slab}
\end{align}
and the components of the electron current density Eq. \eqref{eq:jgy} become
\begin{align}
	J_{x,gc,e} &=  \frac{1}{\Bpstar} \left(\jgcpar B_x + \rhogc E_y \right) \label{eq:jgcx}\\
	J_{y,gc,e} &=  \frac{1}{\Bpstar} \left(\jgcpar B_y - \rhogc E_x \right) \label{eq:jgcy}\\
	J_{z,gc,e} &=  \jgcpar, \label{eq:jgcz}
\end{align}
where we define
\begin{align}
	\rhogc &= q_e \int f_e \Bpstar \dd \vpar , \\
	\jgcpar &= q_e \int \vpar  f_e \Bpstar \dd \vpar.
\end{align}

We then linearize the drift-kinetic Vlasov equation around an equilibrium field with $B_{ext}=B_0$ a given constant and all the other components of $\bE$ and $\bB$ vanishing. The equilibrium distribution function is the Maxwellian Eq. \eqref{eq:Maxwellian} with no drift, constant density and temperature, \textit{i.e.}
\begin{equation}\label{eq:EqDistribution}
	f_{0}(v_\shortparallel)=\frac{n_{0}}{\sqrt{2\pi} v_{th}}\exp^{-\frac{v_\shortparallel^2}{2v_{th}^2}}.
\end{equation}
which does not depend on $\bx$.
Denoting by $f_1,\bE,\bB$ the perturbations, the corresponding linearized drift-kinetic model reads
\begin{equation}
	\fracp{f_1}{t} + \vpar\fracp{f_1}{z} = -\frac{q}{m} E_z\fracp{f_0}{\vpar} .
\end{equation}
Then, multiplying the Ampere equation by $\mu_0$, taking the time derivative and plugging in the expression we get for $\fracp{\bB}{t}$ from Faraday's equation we get 
\begin{equation}
	\frac{1}{c^2}\frac{\partial^2\bE}{\partial t^2} + \frac{1}{V_{A,e}^2} \frac{\partial^2\bE}{\partial t^2} + \nabla\times \nabla\times \bE = -\mu_0\fracp{\bJ_{gc,1}}{t}
\end{equation}
where $ V_{A,e}^2 = {|\Bex|}^2/{{\mu_0  m_e n_0}}$, the perturbed gyrocenter current is obtained by integrating Eqs. \eqref{eq:jgcx}--\eqref{eq:jgcz} keeping only the background density
\begin{equation}
	\bJ_{gc,1}	= \begin{pmatrix}
		\frac{q_e n_0 E_y}{B_0} \\ -\frac{q_e n_0 E_x}{B_0} \\ J_{\shortparallel,1}
	\end{pmatrix}
\end{equation}
denoting by $J_{\shortparallel,1} = q_e \int \vpar f_1 B_0 \dd \vpar$.
As usual, we assume plane wave solutions of the type $\hat{\bE}\exp(\ii (\boldsymbol{k}\cdot\bx - \omega t))$, and the same for the other quantities,  from which we get, assuming without loss of generality that $\boldsymbol{k}=(\kperp,0,\kpar)$
\begin{align}
\ii(\omega - \kpar \vpar) \hat{f}&= \frac{q_e}{m_e}\hat{E}_z \fracp{f_0}{\vpar} \\
 \omega^2 \left(\frac{1}{c^2} + \frac{1}{V_{A,e}^2} \right)\hat{E}_x - \kpar^2 \hat{E}_x + \kpar\kperp\hat{E}_z  &=  -\ii \omega q_e \mu_0 n_0 \hat{E}_y/B_0 = -\ii \frac{\omega \omega_{ce}}{V_{A,e}^2}  \hat{E}_y \\
 \omega^2 \left(\frac{1}{c^2} + \frac{1}{V_{A,e}^2} \right)\hat{E}_y - (\kpar^2+\kperp^2) \hat{E}_y  &= \ii \omega q_e \mu_0 n_0 \hat{E}_x/B_0 = \ii \frac{\omega \omega_{ce}}{V_{A,e}^2}  \hat{E}_x\\
 \frac{\omega^2}{c^2} \hat{E}_z + \kpar\kperp \hat{E}_x - \kperp^2\hat{E}_z  &= -\ii\omega\mu_0\hat{J}_{\shortparallel,1} \label{eq:J1}
\end{align}
using in particular from Faraday's equation that $\omega\hat{B}_z = \kperp\hat{E}_y$, and $\omega_{ce}=q_eB_0/m_e$ is the electron cyclotron frequency. The last step is to express $\hat{J}_{\shortparallel,1}$ from $\hat{f}$ using the first equation above:
\begin{equation}
	\hat{J}_{\shortparallel,1} = q_e B_0\int \vpar \hat{f} \dd \vpar  = -\ii \frac{q_e^2 n_0 B_0}{m_e}\hat{E}_z \int  \frac{\vpar \fracp{f_0(\vpar)}{\vpar}}{\omega - \kpar \vpar} \dd \vpar 
\end{equation}
These integrals can be expressed with respect to the plasma dispersion function
\begin{equation}
	Z(\zeta) = \frac{1}{\sqrt{\pi}}\int \frac{\ee^{-u^2}}{u - \zeta}\dd u,
\end{equation}
which verifies
\begin{equation}
	Z'(\zeta) = \frac{1}{\sqrt{\pi}}\int \frac{\ee^{-u^2}}{(u - \zeta)^2}\dd u = -2( 1 + \zeta Z(\zeta)).
\end{equation}
This yields
\begin{equation}	
	-\ii\omega\mu_0\hat{J}_{\shortparallel,1} = \frac{\omega^2\omega_{pe}^2}{\kpar^2v_{th}^2c^2} \left( 1 + \frac{\omega}{\sqrt{2} \kpar  v_{th}} Z(\frac{\omega}{\sqrt{2} \kpar  v_{th}})\right) \hat{E}_z
\end{equation}
introducing the electron plasma frequency $\omega_{pe}=\sqrt{n_0 q_e^2/(\epsilon_0m_e)}$. This expression can then be plugged into Eq. \eqref{eq:J1}.

Finally, after multiplying the equations above by $\frac{c^2}{\omega^2}$, the terms of the $3\times 3$ dispersion matrix such that the dispersion relation is $D(\boldsymbol{k},\omega)\hat{\bE}=0$
write
\begin{align}
	D_{xx} &= \left(1+\frac{c^2}{V_{A,e}^2}\right) - \frac{c^2}{\omega^2}\kpar^2 \\
	D_{xy} &= -D_{yx} =  i \frac{ q_e n_M}{\epsilon_0  B_{ext}\omega} = i \frac{c^2 \omega_{ce}}{V_{A,e}^2 \omega} \\
	D_{xz} &= D_{zx} =  \frac{c^2}{\omega^2} \kpar k_\perp \\
	D_{yy} &= \left(1+\frac{c^2}{V_{A,e}^2}\right) - \frac{c^2}{\omega^2}(\kpar^2+\kperp^2) \\
	D_{yz} &= -D_{zy} = 0\\
        D_{zz} &= 1+ \frac{\omega_{pe}^2}{\kpar^2v_{th,e}^2} \left[1+\zeta_e Z(\zeta_e)\right] - \frac{c^2}{\omega^2} k_\perp^2   \label{eq:Dzz}
\end{align}
where $\zeta_e=\frac{\omega}{\sqrt{2} \kpar  v_{th,e}}$. This corresponds to the equation $D(\boldsymbol{k},\omega)\bE= (1+\boldsymbol{\chi})\bE + \frac{c^2}{\omega^2} \boldsymbol{k}\times\boldsymbol{k}\times \bE=0$. From this equation, the susceptibility tensor for drift-kinetic electrons can be readily extracted. When considering additional species, it suffices to add $\boldsymbol{\chi}_i$ to $D(\boldsymbol{k},\omega)$ and $\boldsymbol{\chi}=\sum_s \boldsymbol{\chi}_s$. 
\begin{remark}
	A more general dispersion relation for this model keeping all the terms, in particular the nonlinear polarization and magnetization terms which lead to additional $Z$ functions in the dispersion relation has been derived by Zonta et al. \cite{Zonta2021Dispersion}.
\end{remark}

The susceptibility tensor derived from kinetic theory can be found in \cite{brambilla1998kinetic}. For simplicity, we derived the dispersion relation for the coupled drift-kinetic electrons and fully kinetic ions model in the cold plasma limit in Appendix A.

\section{\label{sec:results}Simulation Results}

\subsection{Verification of drift-kinetic electrons}
To verify the drift-kinetic model, we first test a one-species simulation with only electrons with $v_{th,e} =c$.
When $k_\perp=0$, there are three eigenmodes. When $D_{zz}=0$, $E_z$ can be non-zero. It is the Langmuir wave, electrostatic perturbation is parallel to $\Bex$ and parallel propagating. We initialized a density perturbation with $\rho=1+0.04 \cos(k_z z)$. The perturbation is electrostatic ($\bB \ll \bE$) and parallel to $\Bex$. As shown in Fig.~\ref{fig:landau} with $k_z=0.4$, the damping of $E_z$ is observed, and the results are in good agreement with the analytical solution. We also fit the numerical results to determine the frequency and damping rate of the mode for varying $k_z$ as shown by the dots in Fig.~\ref{fig:scan_kpar}. Finally, we show Fig. \ref{fig:ele_disp} the wave spectrum.
The analytical results, obtained by setting $D_{xx}D_{yy}-D_{yx}D_{xy}=0$ using the expression given in Eq. \eqref{eq:Dzz}, are displayed as lines. 
\begin{figure}[htbp]
    \centering
    \begin{subfigure}[b]{0.33\textwidth}
    \includegraphics[width=\linewidth]{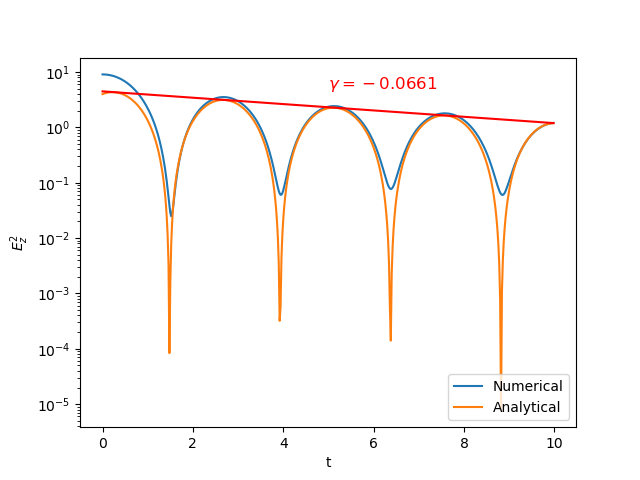} 
    \caption{Electrostatic perturbation.} 
    \label{fig:landau} 
\end{subfigure}
\begin{subfigure}[b]{0.33\textwidth}
    \centering
    \includegraphics[width=\linewidth]{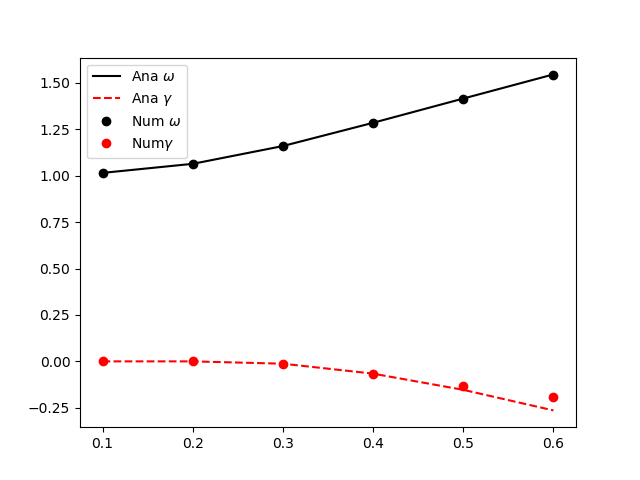} 
    \caption{Scan in $k_z$.} 
    \label{fig:scan_kpar} 
\end{subfigure}
\begin{subfigure}[b]{0.32\textwidth}
    \centering
    \includegraphics[width=\linewidth]{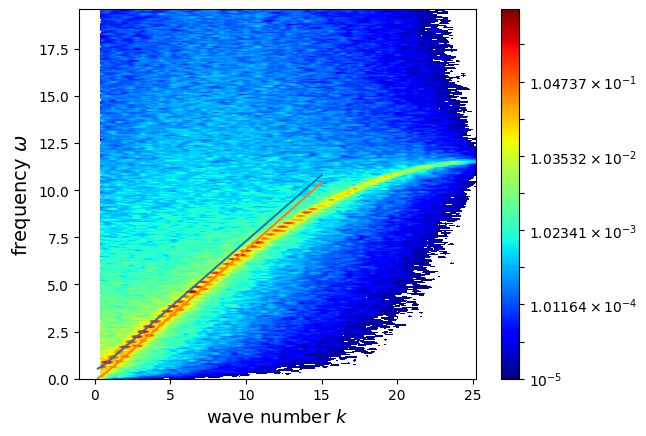} 
    \caption{Wave spectrum} 
    \label{fig:ele_disp} 
\end{subfigure}
\caption{Simulations with only drift-kinetic electrons}
\label{fig:damping}
\end{figure}
We observe two branches of electromagnetic waves propagating along the magnetic field.
The simulation is initialized by the particle noise. 

\subsection{Waves in the Vlasov--Maxwell system} 
To compare the three models: fully kinetic (FK) for both electrons and ions, drift-kinetic electrons fully kinetic ions (Hybrid), and drift-kinetic (DK) for both species, let us consider a two-species simulation with reduced mass ratio $m_i/m_e = 10 $ and $q_i = - q_e$. The initial conditions are Gaussian distributions  with $v_{th,e} = 0.05c$ and $v_{th,i}=0.05c/\sqrt{10}$. We apply a constant exterior field $\Bex$ along the $z$-axis. We use $\omega_{pe}/|\omega_{ce}|=1$. The plasma frequency $\omega_p=\sqrt{\omega_{pe}^2+\omega_{pi}^2}\approx\omega_{pe}$.
To write the dispersion relation and frequencies non-trivially, we assume $m_e\ll m_i$ which yields
\begin{align*}
\omega_{LH}^2 & = {\frac{|\omega_{ce}|  \omega_{ci}}{1 + {\omega_{ce}^2}/{\omega_{pe}^2}}}, ~~~~~ 
\omega_{UH}^2  = {\omega_{pe}^2 + \omega_{ce}^2} \\
\omega_L & = \frac{|\omega_{ce}|}{2} \left( \sqrt{1 + 4\left(\frac{\omega_{pe}}{\omega_{ce}}\right)^2} - 1 \right), ~~~~~
\omega_R  = \frac{|\omega_{ce}|}{2} \left( \sqrt{1 + 4\left(\frac{\omega_{pe}}{\omega_{ce}}\right)^2} + 1 \right),
\end{align*}
where $\omega_{LH}$ and $\omega_{UH}$ are the lower and upper hybrid resonance frequencies, $\omega_{L}$ and $\omega_R$ are the L- and R-wave cutoff frequencies. The frequencies related to the cutoffs and resonance are shown as black dashed lines in the Figs. \ref{fig:fig2} and \ref{fig:fig3} to help identify each branch of the wave. We can find the dispersion relation of the corresponding  modes in the literature \cite{stix1992waves,swanson2020plasma,chen1987waves}.
The O-mode propagates in the direction perpendicular to magnetic field ($\bk \perp \Bex$), with the electric field component $E_z$ parallel to $\Bex$. The X-mode propagates in the direction perpendicular to magnetic field and has two components $E_x$ and $ E_y$  perpendicular to the magnetic field. Their dispersion relations are
$$
\omega^2 = \omega_{p}^2 +c^2k^2 \quad \text{for the ordinary modes (O-mode),}
$$
 and

\[
k^2 = \frac{\omega^2}{c^2} \left[1 - \frac{\omega_{p}^2 \left( \omega^2 - \omega_{p}^2 \right)}{ \omega^2 \left( \omega^2 - \omega_{UH}^2 \right)}\right] \quad \text{for the extraordinary modes (X-mode).}
\]

The Langmuir wave and the L-mode and R-mode propagate in the direction of magnetic field ($\bk \parallel \Bex$). Their dispersion relations are 
$$
\omega^2 = \omega_{p}^2 + 3 k_\parallel^2 v_{the}^2 \quad \text{for the Langmuir wave,}
$$

\[
k^2 = \frac{\omega^2}{c^2} \left[ 1 - \frac{\omega_{p}^2}{(\omega + |\omega_{ce}|)(\omega - \omega_{ci})} \right] \quad \text{for the L-mode,}
\]
 and 

\[
k^2 = \frac{\omega^2}{c^2} \left[ 1 - \frac{\omega_{p}^2}{(\omega - |\omega_{ce}|)(\omega + \omega_{ci})} \right] \quad \text{for the R-mode.}
\]

\subsection{Waves with $\bk$ perpendicular to $\Bex$}
We consider a quasi-one-dimensional simulation with a domain of size 
$
[0, 64 \, d_e] \times [0, d_e] \times [0, d_e]
$ and  a grid of \( 256 \times 8 \times 8 \) points,
where \( d_e = c/\omega_{pe} \) is the electron inertial length.
We use \( 500 \) particles per cell for both species, generated by the quasi-random Sobol sampler. The particle B-spline is of degree  2 in $x$, $y$, and $z$ directions. We employ the 5-stage fourth-order LSRK method with a time step of \( \Delta t = 0.05 \, \omega_{pe}^{-1} \), and the total simulation time is \( T = 200 \, \omega_{pe}^{-1} \).
Figure \ref{fig:fig2} shows the numerical dispersion relation along the $x$ axis of $E_x$, $E_y$, $E_z$ (averaged over $y, z$) for our different models.

The dashed lines representing the analytical results in the cold plasma limitation are obtained by solving the dispersion relations for each model, as detailed in Appendix A.
As shown in Figs. \ref{fig:fig2a}, \ref{fig:fig2d} and \ref{fig:fig2g}, the wave spectrum aligns closely with the analytical results for the X-mode, CAW-X mode, and O-mode. At higher wave numbers, the accuracy of the numerical dispersion relation compared to the analytical results can be further improved by using a higher grid resolution.
The lower X-mode asymptotes to the upper hybrid resonance. When comparing the FK and Hybrid models, the upper X-mode is absent in the Hybrid model. The X-mode in the Hybrid model has a different dispersion relation and cutoff. When $\omega_{pe}/\omega_{ce}$ is smaller (low density) and $m_i/m_e$ is larger, the cutoff frequency approximates to $\omega_L$. And the X-mode in Hybrid mode has no resonance at $\omega_{UH}$, as the electron cyclotron wave is absent.
The transition of compressional Alfv\'en waves (CAW) to the X-mode branch is identical in both the FK and Hybrid models.
In the DK model, even fewer modes exist as shown in the Figs. \ref{fig:fig2c} and \ref{fig:fig2f}.
The CAW does not have resonance at $\omega_{LH}$ as the ion cyclotron effect is absent.
The O-mode, which has a cutoff frequency at $\omega_p$, exists in three models as shown in Figs. \ref{fig:fig2g}, \ref{fig:fig2h} and \ref{fig:fig2i}.
The horizontal lines in the spectrum plot of the FK model corresponding to integer values are the electron Bernstein waves since $\omega_{ce}=1$.
\begin{figure}[htbp]
    \centering
    \begin{subfigure}[b]{0.32\textwidth}
        \centering
        \includegraphics[width=\linewidth]{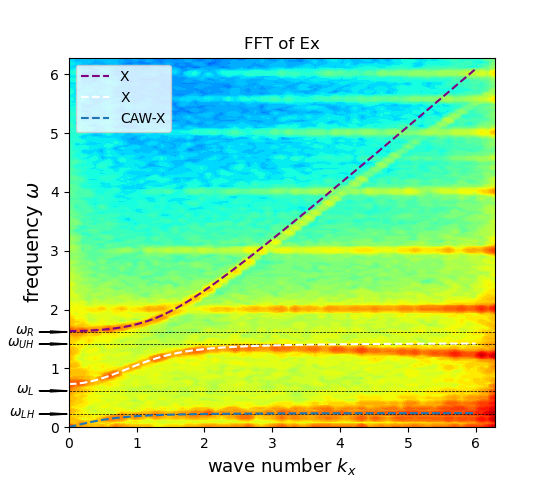}
        \caption{FK. Wave spectrum of $E_x$ vs $k_x$.}
        \label{fig:fig2a}
    \end{subfigure}
    \hspace{2pt} 
    \begin{subfigure}[b]{0.32\textwidth}
        \centering
        \includegraphics[width=\linewidth]{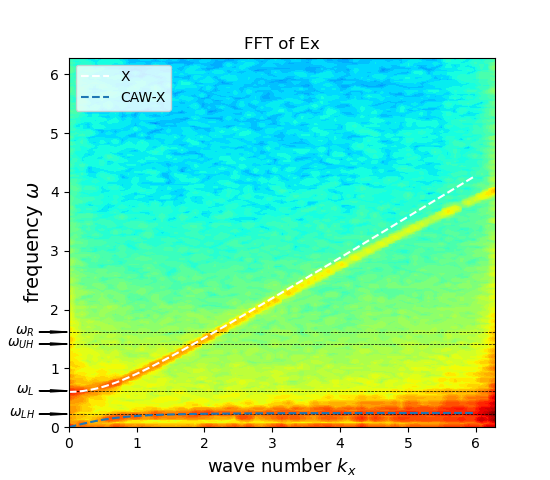}
        \caption{Hybrid. $E_x$ vs $k_x$.}
        \label{fig:fig2b}
    \end{subfigure}
    \hspace{2pt} 
    \begin{subfigure}[b]{0.32\textwidth}
        \centering
        \includegraphics[width=\linewidth]{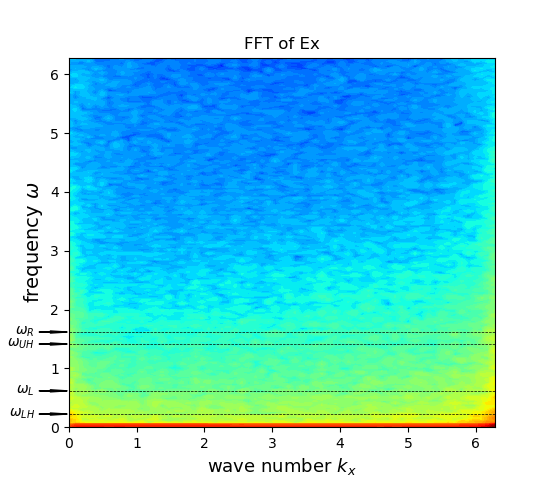}
        \caption{DK. Wave spectrum of $E_x$ vs $k_x$.}
        \label{fig:fig2c}
    \end{subfigure}

    \begin{subfigure}[b]{0.32\textwidth}
        \centering
        \includegraphics[width=\linewidth]{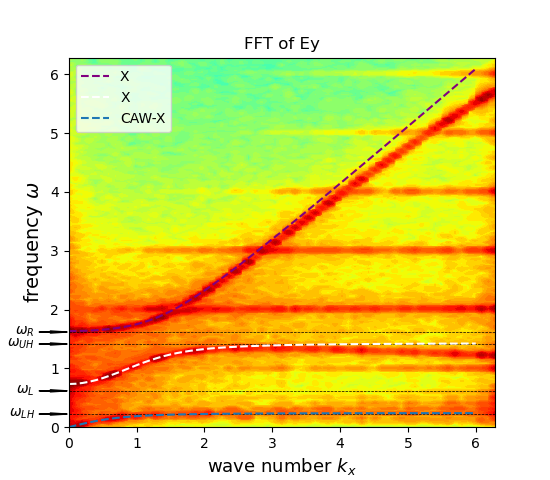}
        \caption{FK. Wave spectrum of $E_y$ vs $k_x$.}
        \label{fig:fig2d}
    \end{subfigure}
    \hspace{2pt} 
    \begin{subfigure}[b]{0.32\textwidth}
        \centering
        \includegraphics[width=\linewidth]{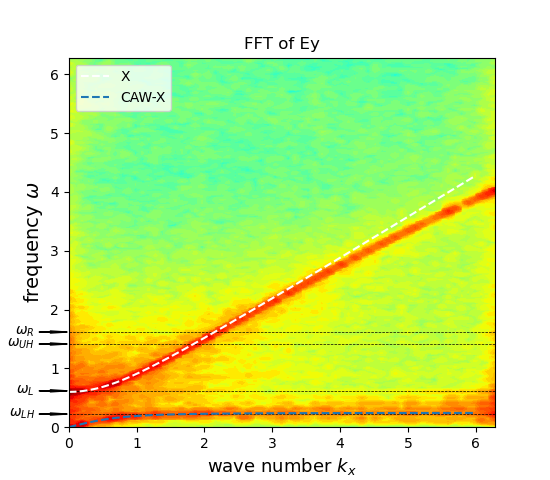}
        \caption{Hybrid. $E_y$ vs $k_x$.}
        \label{fig:fig2e}
    \end{subfigure}
    \hspace{2pt} 
    \begin{subfigure}[b]{0.32\textwidth}
        \centering
        \includegraphics[width=\linewidth]{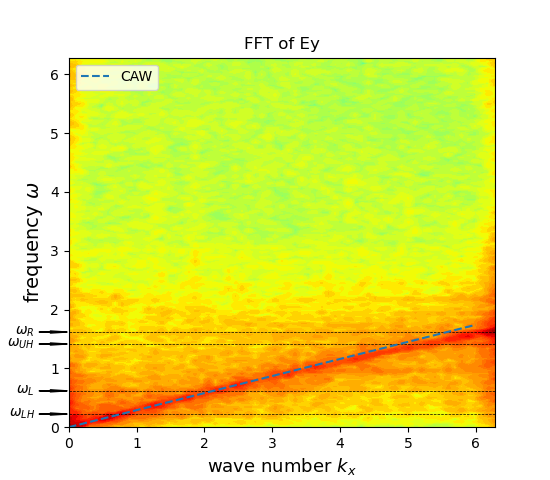}
        \caption{DK. Wave spectrum of $E_y$ vs $k_x$.}
        \label{fig:fig2f}
    \end{subfigure}

    \begin{subfigure}[b]{0.32\textwidth}
        \centering
        \includegraphics[width=\linewidth]{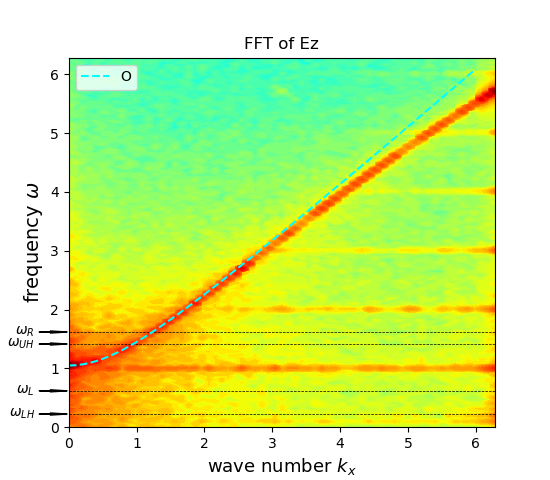}
        \caption{FK. Wave spectrum of $E_z$ vs $k_x$.}
        \label{fig:fig2g}
    \end{subfigure}
    \hspace{2pt} 
    \begin{subfigure}[b]{0.32\textwidth}
        \centering
        \includegraphics[width=\linewidth]{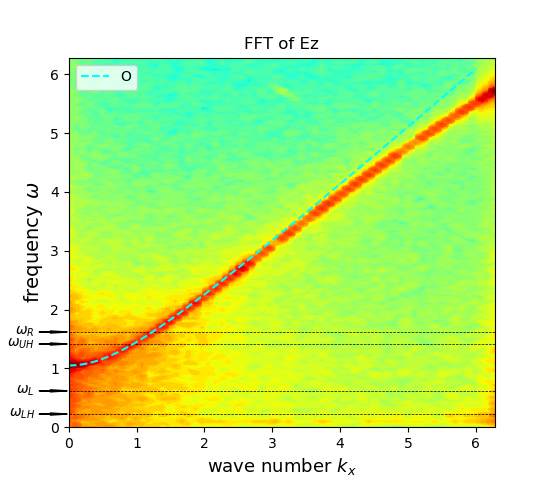}
        \caption{Hybrid. $E_z$ vs $k_x$.}
        \label{fig:fig2h}
    \end{subfigure}
    \hspace{2pt} 
    \begin{subfigure}[b]{0.32\textwidth}
        \centering
        \includegraphics[width=\linewidth]{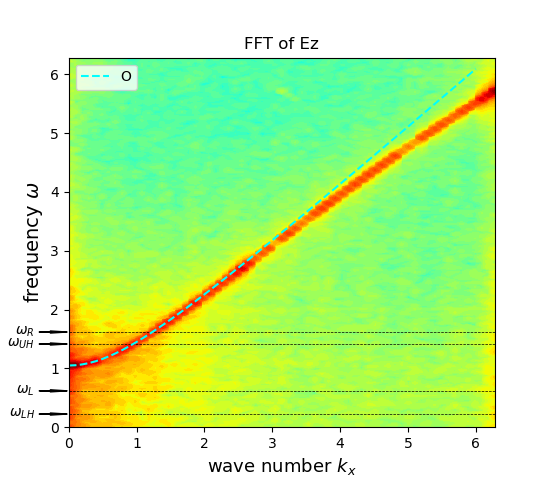}
        \caption{DK. Wave spectrum of $E_z$ vs $k_x$.}
        \label{fig:fig2i}
    \end{subfigure}

    \caption{Comparison of waves with $k$ perpendicular to $\Bex$ for Fully Kinetic and Hybrid models. 
    Left: Fully kinetic for both electrons and ions (FK). 
    Middle: Drift-kinetic electrons with fully kinetic ions (Hybrid). 
    Right: Drift-kinetic for both electrons and ions (DK).}
    \label{fig:fig2}
\end{figure}

\subsection{Waves with $\bk$ parallel to $\Bex$}
We use a domain of size $[0, d_e] \times [0, d_e] \times [0, 64 \, d_e]$ and  a grid of \( 8 \times 8 \times 256 \) points, other parameters are the same. Then we show the numerical dispersion relation along the $z$ axis. The results in the $x$-direction and 
$y$-direction are identical; therefore, we omit the $E_x$-direction wave spectrum in this analysis.

As shown in Fig. \ref{fig:fig3}, the upper R-mode is absent in the Hybrid model compared to the FK model. In Fig. \ref{fig:fig3a}, CAW denotes the compressional Alfv\'en waves and ECW and ICW denote electron and ion cyclotron waves. In the Hybrid model, the CAW does not transit to ECW due to the absence of electron cyclotron resonance as shown in Figs. \ref{fig:fig3a} and \ref{fig:fig3b}. The compressional Alfv\'en waves (CAW) exists resonance near $\omega_{LH}$ at $k$ perpendicular to $B$ and at $\omega_{ci}$ at $k$ parallel to $B$ as shown in Figs. \ref{fig:fig2a}, \ref{fig:fig2d} and Fig. \ref{fig:fig3a}. The CAW-ICW branch is the same in the FK and Hybrid models.
In the DK model, even fewer modes exist as shown in the Fig. \ref{fig:fig3c} and the CAW does not have resonance.
The waves with oscillation in $E_z$ are the same in the three models as shown in Figs. \ref{fig:fig3d}, \ref{fig:fig3e} and \ref{fig:fig3f}. The more accurate dispersion relation for the Langmuir wave are $D_{zz}=0$ in Eq. \eqref{eq:Dzz}, which are damping modes and the damping rate is stronger when $k$ larger as shown in Fig. \ref{fig:scan_kpar}. The resonance frequency in the cold plasma limitation is at $\omega_p$.

As shown in Figs. \ref{fig:fig2} and \ref{fig:fig3}, the Hybrid model is suitable to study the ion cyclotron frequency and low-hybrid waves without modification. When applying the Hybrid model to investigate the lower X-mode, L-mode and upper CAW branches,  the applicable regime should be carefully considered. The DK model is suitable for the low-frequency CAW waves.

\begin{figure}[htbp]
    \centering
    \begin{subfigure}[b]{0.32\textwidth}
        \centering
        \includegraphics[width=\linewidth]{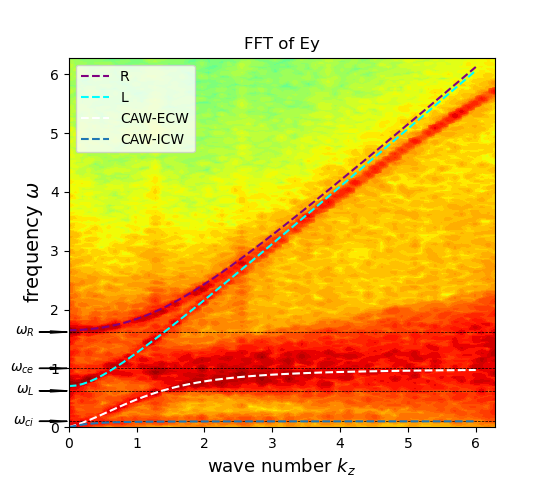}
        \caption{FK. Wave spectrum of $E_y$ vs $k_z$.}
        \label{fig:fig3a}
    \end{subfigure}
    \hspace{2pt} 
    \begin{subfigure}[b]{0.32\textwidth}
        \centering
        \includegraphics[width=\linewidth]{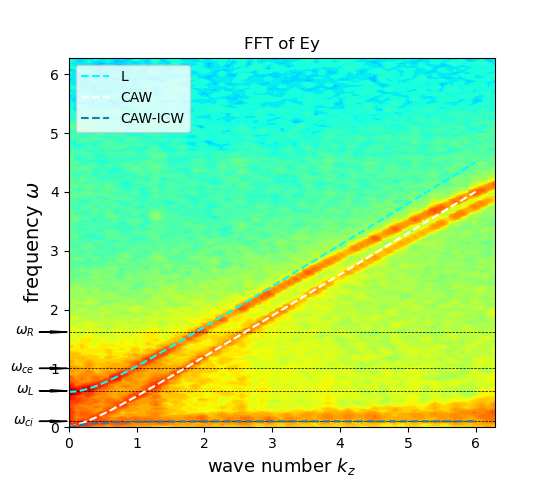}
        \caption{Hybrid. $E_y$ vs $k_z$.}
        \label{fig:fig3b}
    \end{subfigure}
    \hspace{2pt} 
    \begin{subfigure}[b]{0.32\textwidth}
        \centering
        \includegraphics[width=\linewidth]{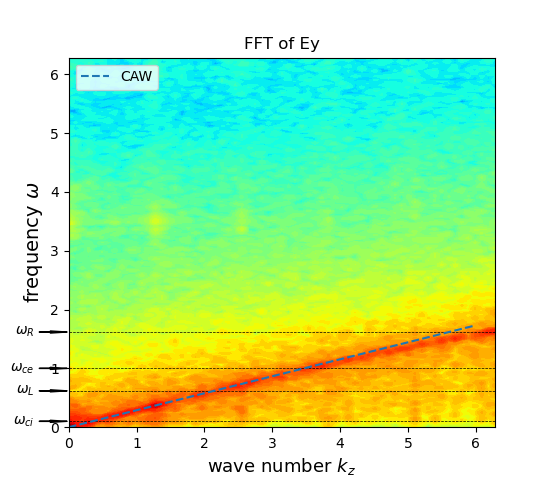}
        \caption{DK. Wave spectrum of $E_y$ vs $k_z$.}
        \label{fig:fig3c}
    \end{subfigure}

    \begin{subfigure}[b]{0.32\textwidth}
        \centering
        \includegraphics[width=\linewidth]{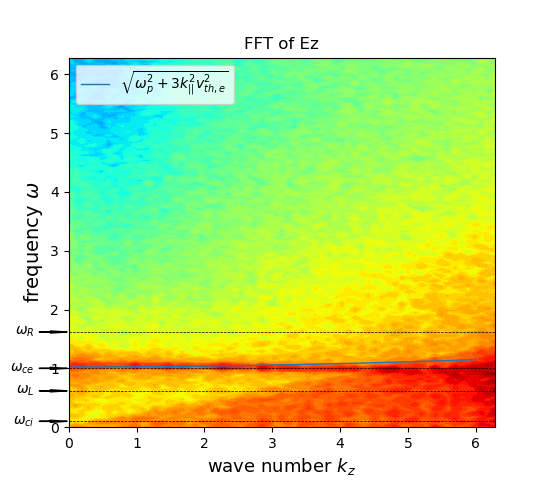}
        \caption{FK. Wave spectrum of $E_z$ vs $k_z$.}
        \label{fig:fig3d}
    \end{subfigure}
    \hspace{2pt} 
    \begin{subfigure}[b]{0.32\textwidth}
        \centering
        \includegraphics[width=\linewidth]{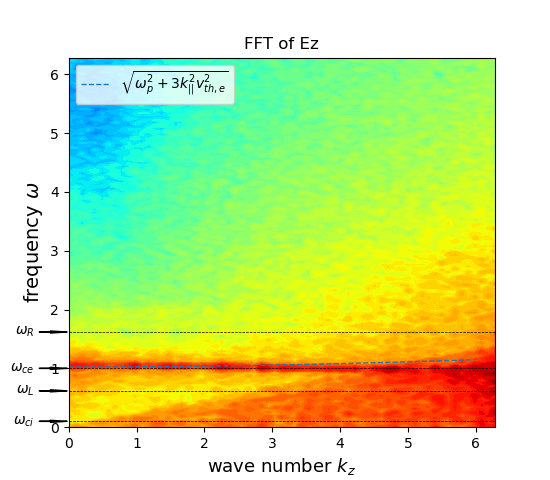}
        \caption{Hybrid. $E_z$ vs $k_z$.}
        \label{fig:fig3e}
    \end{subfigure}
    \hspace{2pt} 
    \begin{subfigure}[b]{0.32\textwidth}
        \centering
        \includegraphics[width=\linewidth]{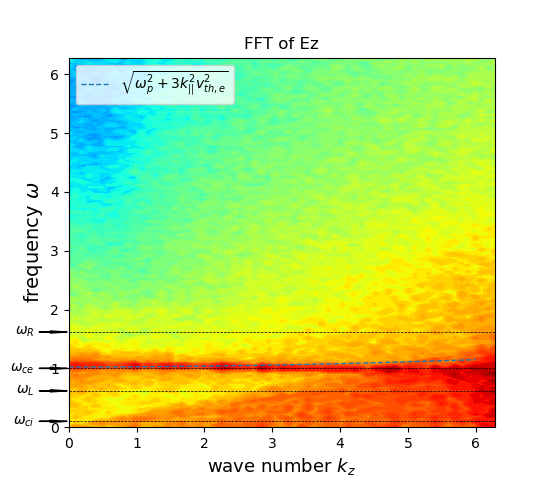}
        \caption{DK. Wave spectrum of $E_z$ vs $k_z$.}
        \label{fig:fig3f}
    \end{subfigure}
    \caption{Compare the waves with $k$ parallel to $\Bex$ for the Fully Kinetic and Hybrid models.
    Left: Fully kinetic for both electrons and ions (FK). 
    Middle: Drift-kinetic electrons with fully kinetic ions (Hybrid). 
    Right: Drift-kinetic for both electrons and ions (DK).}
    \label{fig:fig3}
\end{figure}

\section{\label{sec:conclusion}Conclusions and outlook}

We have derived a new geometric PIC discretization for a gauge-free drift-kinetic model that can directly be combined with a fully kinetic model. As both use only the physical electric and magnetic field and not the potentials, the coupling with the fields can be written as a macroscopic Maxwell system including polarization and magnetization effects coming from the drift-kinetic particles. The next steps that will be addressed in a future work are to include quasi-neutrality in the model. In this case the time derivative of the electric field does not appear anymore in Ampere's law, so that new models will be needed to solve for the parallel electric field, when the ion polarization term remains and the full electric field when this is not present, which is the case when the ions are treated with the full kinetic model.

The model that we derived is fully general and can be applied for real tokamak or stellarator geometry. However, we only implemented and tested it for the moment in slab geometry as discussed in Sections 6 \& 7.
The next step will be to implement the missing terms related to a non constant background magnetic field, aiming applications to edge physics.
Regarding multiscale physics, our formulation naturally integrates different kinetic models within a single consistent framework. The macroscopic Maxwell system in our formulation including polarization and magnetization effects from the drift-kinetic particles, enabling seamless coupling between drift-kinetic and fully kinetic treatments without relying on moment closures. This makes our approach particularly well-suited for studying cross-scale interactions, such as turbulence-driven transport, kinetic instabilities, and wave-particle interactions, which are crucial in edge and core plasma dynamics. We plan to further explore these applications in future studies.

\section*{Acknowledgments}
This work is part of the Eurofusion project TSVV-4.
This work has been carried out within the framework of the EUROfusion Consortium, funded by the European Union via the Euratom Research and Training Programme (Grant Agreement No 101052200 - EUROfusion). Views and opinions expressed are however those of the author(s) only and do not necessarily reflect those of the European Union or the European Commission. Neither the European Union nor the European Commission can be held responsible for them.

\section*{Appendix A. Dispersion relation of the three models in the cold plasma limitation}\label{ap:CPDR}
The dispersion relation is written as $D(\boldsymbol{k},\omega)\bE= (1+\boldsymbol{\chi})\bE + \frac{c^2}{\omega^2} \boldsymbol{k}\times\boldsymbol{k}\times \bE=0$.
 For the well known cold plasma dispersion relation (CPDR):
\[
\overleftrightarrow{\mathbf{D}}(\boldsymbol{k}, \omega) = 
\begin{bmatrix}
S - n^2 \cos^2\theta & -iD & n^2 \sin\theta \cos\theta \\
iD & S - n^2 & 0 \\
n^2 \sin\theta \cos\theta & 0 & P - n^2 \sin^2\theta
\end{bmatrix}, 
\]
where $n \equiv {c {k}}/{\omega}$, the wave vector $\boldsymbol{k} = (k \sin\theta, 0, k \cos\theta)$.

The corresponding quantities in Stix notation \cite{stix1992waves} are:
\[
S = 1 - \sum_s \frac{\omega_{ps}^2}{\omega^2 - \omega_{cs}^2}, ~~~
D = \sum_s \frac{\omega_{cs} \omega_{ps}^2}{\omega (\omega^2 - \omega_{cs}^2)},
~~~
P = 1 - \sum_s \frac{\omega_{ps}^2}{\omega^2}.
\]
We can replace $\chi_e$ with the $\chi_e$ derived using the drift-kinetic model. When there is only one type of ion in the system besides the electrons, then we can derive in the cold plasma dispersion relation for the drift-kinetic electron fully-kinetic ion (Hybrid) model for which
\[
S = 1 + \frac{\omega_{pe}^2}{\omega_{ce}^2} - \frac{\omega_{pi}^2}{\omega^2 - \omega_{ci}^2}, 
~~~
D = -\frac{\omega_{pe}^2}{\omega \omega_{ce}}+\frac{\omega_{ci} \omega_{pi}^2}{\omega (\omega^2 - \omega_{ci}^2)},
~~~
P = 1 - \frac{\omega_{p}^2}{\omega^2},
\]
where $\omega_{p}^2=\omega_{pe}^2+\omega_{pi}^2$. For the drift-kinetic electrons, there is no resonance at the electron cyclone frequency. 
And for the drift-kinetic model for both electrons and ions (DK),
\[
S = 1 +  \sum_s \frac{\omega_{ps}^2}{\omega_{cs}^2}, 
~~~
D = -\sum_s \frac{\omega_{ps}^2}{\omega \omega_{cs}},
~~~
P = 1 - \frac{\omega_{p}^2}{\omega^2}.
\]
In the equations for $S$, $D$ and $P$, the cyclotron frequency is defined as $\omega_{cs} \equiv q_s B_{\text{ext}} / m_s$. Note that $q_s$ can be either positive or negative. The dispersion relation can be expressed as a polynomial equation, for which established methods can be used to determine all the roots numerically \cite{xie2019bo}. The polynomial function for CPDR \cite{swanson2020plasma} can be written as
\begin{equation}
    c_{10}\omega^{10} - c_8\omega^8 + c_6\omega^6 - c_4\omega^4 + c_2\omega^2 - c_0 = 0,
\end{equation}
where
\begin{align*}
c_0 &= c^4 k^4 \omega_{ce}^4 \omega_{ci}^4 \omega_p^2 \cos^2\theta, \\
c_2 &= c^4 k^4 \big[\omega_p^2 (\omega_{ce}^2 + \omega_{ci}^2 - \omega_{ci} \omega_{ce}) \cos^2\theta 
+ \omega_{ci} \omega_{ce} (\omega_p^2 + \omega_{ci} \omega_{ce}) \big] \\
    \quad & + c^2 k^2 \omega_p^2 \omega_{ci} \omega_{ce} (\omega_p^2 + \omega_{ci} \omega_{ce}) (1 + \cos^2\theta), \\
c_4 &= c^4 k^4 (\omega_{ce}^2 + \omega_{ci}^2 + \omega_p^2) + 2c^2 k^2 (\omega_p^2 + \omega_{ci} \omega_{ce})^2\\
 \quad & +c^2 k^2 \omega_p^2 (\omega_{ce}^2 + \omega_{ci}^2 - \omega_{ci} \omega_{ce})(1 + \cos^2\theta) 
+ \omega_p^2 (\omega_p^2 + \omega_{ci} \omega_{ce})^2, \\
c_6 &= c^4 k^4 + \big(2c^2 k^2 + \omega_p^2 \big)(\omega_{ce}^2 + \omega_{ci}^2 + 2\omega_p^2) 
+ (\omega_p^2 + \omega_{ci} \omega_{ce})^2, \\
c_{8} &= 2c^2 k^2 + \omega_{ce}^2 + \omega_{ci}^2 + 3\omega_p^2, \\
c_{10} &= 1.
\end{align*}
Note that here, similar to the Eq. (2.63) in \cite{swanson2020plasma}, $\omega_{ce}$ has been modified to represent $|\omega_{ce}|$ as a positive value, and $q_i$ is assumed to be $|q_e|=e$. The condition $\omega_{ci} \omega_{pe}^2 - \omega_{ce} \omega_{pi}^2 = 0$ is used to simplify the result. The polynomial is fifth order in $\omega^2$, which generally has five sets of roots.

We can use the same method to derive the coefficients of the Hybrid and DK models. The same symbols and conditions are used for ease of comparison between the dispersion relation functions of Hybrid, FK and DK models. The resulting polynomial of the Hybrid model is of fourth order in $\omega^2$, 
\begin{equation}
    - c_8\omega^8 + c_6\omega^6 - c_4\omega^4 + c_2\omega^2 - c_0 = 0,
\end{equation}
where
\begin{align*}
c_0 &= c^4 k^4 \omega_{ce}^4 \omega_{ci}^2 \omega_p^2 \cos^2\theta, \\
c_2 &= c^4 k^4 \omega_{ce}^2\big[\omega_{pe}^2 (\omega_{ce}^2 - \omega_{ci}^2 ) \cos^2\theta 
+ \omega_{ci} \omega_{ce} (\omega_p^2 + \omega_{ci} \omega_{ce}) \big] \\
    \quad & + c^2 k^2 \omega_p^2 \omega_{ci} \omega_{ce}^3 (\omega_p^2 + \omega_{ci} \omega_{ce}) (1 + \cos^2\theta), \\
c_4 &= c^4 k^4 \omega_{ce}^2 \big[\omega_{ce}^2 + \omega_{pe}^2 ( 1- \cos^2\theta )\big] + 2c^2 k^2 \omega_{ce}^2 (\omega_p^2 + \omega_{ci} \omega_{ce})^2  \\
\quad & + c^2 k^2 \omega_{ce}^2 \big[\omega_{pe}^2 (\omega_{ce}^2-\omega_{ci}^2 ) - \omega_{pi}^2 \omega_{p}^2 \big] (1+\cos^2\theta) \\
\quad & + \omega_p^2 \omega_{ce}^2 (\omega_p^2 + \omega_{ci} \omega_{ce})^2, \\
c_6 &= c^2 k^2 (\omega_{ce}^2 + \omega_{pe}^2) \left[ 2 \omega_{ce}^2 + \omega_{pe}^2 ( 1- \cos^2\theta ) \right] \\
\quad &+ \left[  \omega_{p}^2 (\omega_{p}^2 + \omega_{ce}^2 )^2 
+ \omega_{ce}^2 (\omega_p^2 + \omega_{ci} \omega_{ce})^2 \right], \\
c_{8} &= (\omega_{ce}^2 + \omega_{pe}^2)^2.
\end{align*}
For the DK model the resulting polynomial is of third order in $\omega^2$,
\begin{equation}
     c_6\omega^6 - c_4\omega^4 + c_2\omega^2 - c_0 = 0,
\end{equation}
where
\begin{align*}
c_0 &= c^4 k^4 \omega_{p}^2\omega_{ce}^2 \omega_{ci}^2  \cos^2\theta, \\
c_2 &= c^4 k^4 \omega_{ce}\omega_{ci} \big[\omega_{ce}\omega_{ci} + \omega_{p}^2 ( 1- \cos^2\theta )\big] ,\\
\quad & + c^2 k^2 \omega_{p}^2  \omega_{ce}  \omega_{ci} (\omega_p^2 + \omega_{ci} \omega_{ce}) ( 1 + \cos^2\theta ) , \\
c_4 &= c^2 k^2 \omega_{p}^2 (\omega_p^2 + \omega_{ci} \omega_{ce}) ( 1- \cos^2\theta ) \\
\quad & + 2c^2 k^2 \omega_{ce}\omega_{ci} (\omega_p^2 + \omega_{ci} \omega_{ce}) + \omega_p^2 (\omega_{p}^2 + \omega_{ce}\omega_{ci})^2, \\
c_{6} &= (\omega_{p}^2 + \omega_{ce}\omega_{ci})^2.
\end{align*}
For a detailed discussion of each wave branch, please refer to \cite{swanson2020plasma,chen1987waves} or other foundational books.

\bibliographystyle{iopart-num}
\providecommand{\newblock}{}

\end{document}